\documentclass[11pt]{article}
\usepackage{typearea}\typearea{13}
\usepackage{amsmath}
\usepackage{amssymb}\usepackage{amsthm}
\usepackage{amsfonts}
\usepackage[dvipdfmx]{hyperref}\hypersetup{hidelinks}
\usepackage{color}
\numberwithin{equation}{section}

\begin{document}
\theoremstyle{plain}
\newtheorem{thm}{Theorem}[section]
\newtheorem{lem}[thm]{Lemma}
\newtheorem{prop}[thm]{Proposition}
\newtheorem{cor}[thm]{Corollary}
\theoremstyle{definition}
\newtheorem{assum}[thm]{Assumption}
\newtheorem{notation}[thm]{Notation}
\newtheorem{defn}[thm]{Definition}
\newtheorem{clm}[thm]{Claim}
\newtheorem{ex}[thm]{Example}
\theoremstyle{remark}
\newtheorem{rem}[thm]{Remark}
\newcommand{\unit}{\mathbb I}
\newcommand{\ali}[1]{{\mathfrak A}_{[ #1 ,\infty)}}
\newcommand{\alm}[1]{{\mathfrak A}_{(-\infty, #1 ]}}
\newcommand{\nn}[1]{\lV #1 \rV}
\newcommand{\br}{{\mathbb R}}
\newcommand{\dm}{{\rm dom}\mu}
\newcommand{\lb}{l_{\bb}(n,n_0,k_R,k_L,\lal,\bbD,\bbG,Y)}
\newcommand{\Ad}{\mathop{\mathrm{Ad}}\nolimits}
\newcommand{\Proj}{\mathop{\mathrm{Proj}}\nolimits}
\newcommand{\RRe}{\mathop{\mathrm{Re}}\nolimits}
\newcommand{\RIm}{\mathop{\mathrm{Im}}\nolimits}
\newcommand{\Wo}{\mathop{\mathrm{Wo}}\nolimits}
\newcommand{\Prim}{\mathop{\mathrm{Prim}_1}\nolimits}
\newcommand{\Primz}{\mathop{\mathrm{Prim}}\nolimits}
\newcommand{\ClassA}{\mathop{\mathrm{ClassA}}\nolimits}
\newcommand{\Class}{\mathop{\mathrm{Class}}\nolimits}
\def\qed{{\unskip\nobreak\hfil\penalty50
\hskip2em\hbox{}\nobreak\hfil$\square$
\parfillskip=0pt \finalhyphendemerits=0\par}\medskip}
\def\proof{\trivlist \item[\hskip \labelsep{\bf Proof.\ }]}
\def\endproof{\null\hfill\qed\endtrivlist\noindent}
\def\proofof[#1]{\trivlist \item[\hskip \labelsep{\bf Proof of #1.\ }]}
\def\endproofof{\null\hfill\qed\endtrivlist\noindent}
\newcommand{\caA}{{\mathcal A}}
\newcommand{\caB}{{\mathcal B}}
\newcommand{\caC}{{\mathcal C}}
\newcommand{\caD}{{\mathcal D}}
\newcommand{\caE}{{\mathcal E}}
\newcommand{\caF}{{\mathcal F}}
\newcommand{\caG}{{\mathcal G}}
\newcommand{\caH}{{\mathcal H}}
\newcommand{\caI}{{\mathcal I}}
\newcommand{\caJ}{{\mathcal J}}
\newcommand{\caK}{{\mathcal K}}
\newcommand{\caL}{{\mathcal L}}
\newcommand{\caM}{{\mathcal M}}
\newcommand{\caN}{{\mathcal N}}
\newcommand{\caO}{{\mathcal O}}
\newcommand{\caP}{{\mathcal P}}
\newcommand{\caQ}{{\mathcal Q}}
\newcommand{\caR}{{\mathcal R}}
\newcommand{\caS}{{\mathcal S}}
\newcommand{\caT}{{\mathcal T}}
\newcommand{\caU}{{\mathcal U}}
\newcommand{\caV}{{\mathcal V}}
\newcommand{\caW}{{\mathcal W}}
\newcommand{\caX}{{\mathcal X}}
\newcommand{\caY}{{\mathcal Y}}
\newcommand{\caZ}{{\mathcal Z}}
\newcommand{\bbA}{{\mathbb A}}
\newcommand{\bbB}{{\mathbb B}}
\newcommand{\bbC}{{\mathbb C}}
\newcommand{\bbD}{{\mathbb D}}
\newcommand{\bbE}{{\mathbb E}}
\newcommand{\bbF}{{\mathbb F}}
\newcommand{\bbG}{{\mathbb G}}
\newcommand{\bbH}{{\mathbb H}}
\newcommand{\bbI}{{\mathbb I}}
\newcommand{\bbJ}{{\mathbb J}}
\newcommand{\bbK}{{\mathbb K}}
\newcommand{\bbL}{{\mathbb L}}
\newcommand{\bbM}{{\mathbb M}}
\newcommand{\bbN}{{\mathbb N}}
\newcommand{\bbO}{{\mathbb O}}
\newcommand{\bbP}{{\mathbb P}}
\newcommand{\bbQ}{{\mathbb Q}}
\newcommand{\bbR}{{\mathbb R}}
\newcommand{\bbS}{{\mathbb S}}
\newcommand{\bbT}{{\mathbb T}}
\newcommand{\bbU}{{\mathbb U}}
\newcommand{\bbV}{{\mathbb V}}
\newcommand{\bbW}{{\mathbb W}}
\newcommand{\bbX}{{\mathbb X}}
\newcommand{\bbY}{{\mathbb Y}}
\newcommand{\bbZ}{{\mathbb Z}}
\newcommand{\str}{^*}
\newcommand{\lv}{\left \vert}
\newcommand{\rv}{\right \vert}
\newcommand{\lV}{\left \Vert}
\newcommand{\rV}{\right \Vert}
\newcommand{\la}{\left \langle}
\newcommand{\ra}{\right \rangle}
\newcommand{\ltm}{\left \{}
\newcommand{\rtm}{\right \}}
\newcommand{\lcm}{\left [}
\newcommand{\rcm}{\right ]}
\newcommand{\ket}[1]{\lv #1 \ra}
\newcommand{\bra}[1]{\la #1 \rv}
\newcommand{\lmk}{\left (}
\newcommand{\rmk}{\right )}
\newcommand{\al}{{\mathcal A}}
\newcommand{\md}{M_d({\mathbb C})}
\newcommand{\id}{\mathop{\mathrm{id}}\nolimits}
\newcommand{\Tr}{\mathop{\mathrm{Tr}}\nolimits}
\newcommand{\Ran}{\mathop{\mathrm{Ran}}\nolimits}
\newcommand{\Ker}{\mathop{\mathrm{Ker}}\nolimits}
\newcommand{\spn}{\mathop{\mathrm{span}}\nolimits}
\newcommand{\Mat}{\mathop{\mathrm{M}}\nolimits}
\newcommand{\UT}{\mathop{\mathrm{UT}}\nolimits}
\newcommand{\DT}{\mathop{\mathrm{DT}}\nolimits}
\newcommand{\GL}{\mathop{\mathrm{GL}}\nolimits}
\newcommand{\spa}{\mathop{\mathrm{span}}\nolimits}
\newcommand{\supp}{\mathop{\mathrm{supp}}\nolimits}
\newcommand{\rank}{\mathop{\mathrm{rank}}\nolimits}
\newcommand{\idd}{\mathop{\mathrm{id}}\nolimits}
\newcommand{\ran}{\mathop{\mathrm{Ran}}\nolimits}
\newcommand{\dr}{ \mathop{\mathrm{d}_{{\mathbb R}^k}}\nolimits} 
\newcommand{\dc}{ \mathop{\mathrm{d}_{\cc}}\nolimits} \newcommand{\drr}{ \mathop{\mathrm{d}_{\rr}}\nolimits} 
\newcommand{\zin}{\mathbb{Z}}
\newcommand{\rr}{\mathbb{R}}
\newcommand{\cc}{\mathbb{C}}
\newcommand{\ww}{\mathbb{W}}
\newcommand{\nan}{\mathbb{N}}\newcommand{\bb}{\mathbb{B}}
\newcommand{\aaa}{\mathbb{A}}\newcommand{\ee}{\mathbb{E}}
\newcommand{\pp}{\mathbb{P}}
\newcommand{\wks}{\mathop{\mathrm{wk^*-}}\nolimits}
\newcommand{\he}{\hat {\mathbb E}}
\newcommand{\ikn}{{\caI}_{k,n}}
\newcommand{\mk}{{\Mat_k}}
\newcommand{\mnz}{\Mat_{n_0}}
\newcommand{\mn}{\Mat_{n}}
\newcommand{\mkk}{\Mat_{k_R+k_L+1}}
\newcommand{\mnzk}{\mnz\otimes \mkk}
\newcommand{\hbb}{H^{k,\bb}_{m,p,q}}
\newcommand{\gb}[1]{\Gamma^{(R)}_{#1,\bb}}
\newcommand{\cgv}[1]{\caG_{#1,\vv}}
\newcommand{\gv}[1]{\Gamma^{(R)}_{#1,\vv}}
\newcommand{\gvt}[1]{\Gamma^{(R)}_{#1,\vv(t)}}
\newcommand{\gbt}[1]{\Gamma^{(R)}_{#1,\bb(t)}}
\newcommand{\cgb}[1]{\caG_{#1,\bb}}
\newcommand{\cgbt}[1]{\caG_{#1,\bb(t)}}
\newcommand{\gvp}[1]{G_{#1,\vv}}
\newcommand{\gbp}[1]{G_{#1,\bb}}
\newcommand{\gbpt}[1]{G_{#1,\bb(t)}}
\newcommand{\Pbm}[1]{\Phi_{#1,\bb}}
\newcommand{\Pvm}[1]{\Phi_{#1,\bb}}
\newcommand{\mb}{m_{\bb}}
\newcommand{\E}[1]{\widehat{\mathbb{E}}^{(#1)}}
\newcommand{\lal}{{\boldsymbol\lambda}}
\newcommand{\rar}{{\boldsymbol r}}
\newcommand{\oo}{{\boldsymbol\omega}}
\newcommand{\vv}{{\boldsymbol v}}
\newcommand{\bbm}{{\boldsymbol m}}
\newcommand{\braket}[2]{\langle#1,#2\rangle}
\newcommand{\abs}[1]{\left\vert#1\right\vert}
\newtheorem{nota}{Notation}[section]
\def\qed{{\unskip\nobreak\hfil\penalty50
\hskip2em\hbox{}\nobreak\hfil$\square$
\parfillskip=0pt \finalhyphendemerits=0\par}\medskip}
\def\proof{\trivlist \item[\hskip \labelsep{\bf Proof.\ }]}
\def\endproof{\null\hfill\qed\endtrivlist\noindent}
\def\proofof[#1]{\trivlist \item[\hskip \labelsep{\bf Proof of #1.\ }]}
\def\endproofof{\null\hfill\qed\endtrivlist\noindent}
\newcommand{\todo}[1]{{\bf [Todo: #1]}}

\newcommand{\qq}{\mathbb{Q}}
\newcommand{\NN}{\mathbb{N}}
\newcommand{\CC}{\mathbb{C}}
\newcommand{\ZZ}{\mathbb{Z}}

\newcommand{\1}{\mathbbm{1}} 
\newcommand{\spec}[1]{\mathrm{sp}(#1)}

\newcommand{\sig}{\sigma}
\newcommand{\ep}{\varepsilon}
\newcommand{\del}{\delta}
\newcommand{\ga}{\gamma}
\newcommand{\om}{\omega}

\newcommand{\set}[1]{\left\{ #1 \right\} }
\newcommand{\ip}[1]{\langle #1 \rangle}
\newcommand{\norm}[1]{\Vert #1 \Vert }

\newcommand{\m}[1]{\mathbb{#1}}
\newcommand{\mc}[1]{\mathcal{#1}}
\newcommand{\mf}[1]{\mathfrak{#1}}

\newtheorem{theorem}{Theorem}[section]
\newtheorem{lemma}[theorem]{Lemma}
\newtheorem{proposition}[theorem]{Proposition}
\newtheorem{corollary}[theorem]{Corollary}
\newtheorem{definition}[theorem]{Definition}
 
\newcommand{\vphi}{\varphi}
\newcommand{\specc}{\mathrm{sp}}
\newcommand{\third}{\bigg{(}\frac{1}{3} \bigg{)}}
\newcommand{\fnorm}{\norm{\Phi^1(\ep)}_{F_\vphi} }
\newcommand{\floor}[1]{\lfloor #1 \rfloor}
\newcommand{\diam}[1]{\text{diam}(#1)}
\newcommand{\Aloc}{\mathcal{A}_{\text{loc}}}
\newcommand{\ALa}{\mathcal{A}_{\Lambda}}

\newcommand{\A}{\mathcal{A}}
\newcommand{\EE}{\mathbb{E}}

\title{Automorphic equivalence within gapped phases in the bulk}
\author{
{\sc Alvin Moon}\footnote{Supported in part by National Science Foundation Grant DMS 1813149.}\\
{\small Department of Mathematics}\\
{\small University of California, Davis. Davis, CA, 95616 USA} \\
and\\
{\sc Yoshiko Ogata}\footnote{Supported in part by
the Grants-in-Aid for
Scientific Research, JSPS 16K05171.}\\
{\small Graduate School of Mathematical Sciences}\\
{\small The University of Tokyo, Komaba, Tokyo, 153-8914, Japan}
}
\maketitle
\begin{abstract}
We develop a new adiabatic theorem for unique gapped ground states which does not require the gap for local Hamiltonians.
We instead require a gap in the bulk and a smoothness of expectation values of sub-exponentially localized observables in the unique gapped ground state $\varphi_s(A)$.
This requirement is weaker than the requirement of the gap of the local Hamiltonians, since a uniform spectral gap for finite dimensional ground states implies a gap in the bulk for unique gapped ground states, 
as well as the smoothness.

\end{abstract}

\section{Introduction}
Hastings's  \cite{Has} \cite{hw} adiabatic method is a powerful tool in the analysis of 
gapped Hamiltonians in quantum many-body systems. Seminal mathematical developments from \cite{BMNS}, \cite{NSY}, \cite{Y} and onwards have established a strong mathematical framework of adiabatic theory for quantum many-body systems. The adiabatic theorems from these works state that for a smooth path of gapped Hamiltonians, there is an automorphic equivalence between ground state spaces along the path. Furthermore, these automorphisms are quasi-local. 

This framework has proven to be broadly applicable to many situations. In \cite{hm}, the long standing problem of explaining the quantization of the Hall conductance was finally solved with this method. Using the idea in \cite{hm}, the Kubo formula was derived in \cite{bdf}.
 
Another use of the adiabatic theorem is the analysis of symmetry protected topological (SPT) phase, in \cite{tri} and \cite{reflection}. In \cite{tri} and \cite{reflection}, indices for SPT phases which extend the indices by Pollmann et.al. \cite{po},\cite{po2} were introduced. The adiabatic theorem was used to show the stability of these indices. 
See  \cite{Mo} for the extension of \cite{tri}
to interactions with unbounded interaction range with fast decay. 

All of the adiabatic theorems developed so far require a uniform spectral gap for local Hamiltonians.
Therefore, even if what we are interested in is the bulk, the use of known adiabatic theorems requires 
conditions on the gap in finite boxes. This is conceptually unsatisfactory because bulk-classification of gapped Hamiltonians
can be coarser than the classification in finite volume \cite{Ogata3}. 
In this paper, we develop a new adiabatic theorem for unique gapped ground states which does not require the gap for local Hamiltonians.
We instead require a gap in the bulk and a smoothness of expectation values of sub-exponentially localized observables in the unique gapped ground state $\varphi_s(A)$.
This requirement is weaker than the requirement of the gap of the local Hamiltonians, since a uniform spectral gap for finite dimensional ground states implies a gap in the bulk for unique gapped ground states, 
as well as the smoothness. {(See Remark \ref{imply}.)}
%
Under such conditions, we show that there is a smooth path of
quasi-local automorphisms $\alpha_s$, such that $\omega_s=\omega_0\circ \alpha_s$.
This $\alpha_s$ is the same as the one given in the literatures \cite{BMNS}, \cite{NSY}.

Although the result is analogous to those of finite systems, there is a crucial difference for the proof. For 
the finite system $\caA_\Lambda$, there is a Hamiltonian $H_s({\Lambda})$ in the $C^*$-algebra
$\caA_\Lambda$. By considering a differential equation satisfied by the spectral projection $P_s(\Lambda)$ of the Hamiltonian $H_s(\Lambda)$ corresponding
to the lowest eigenvalue, we may explicitly define in this case the automorphisms connecting the ground state spaces. 
In contrast, for infinite systems, we do not have a Hamiltonian $H_s$ in the $C^*$-algebra of quantum spin systems. Of course we can consider the bulk Hamiltonian $H_s$, but $H_s$ depends on the GNS representation, and the meaning of $\frac{d}{ds} H_s$ is ambiguous. Therefore, we have to find an alternative way to prove our adiabatic theorem.

%

Let us now give a more precise description of our result.
We start by summarizing the standard setup of quantum spin systems \cite{BR1,BR2}.
Let $\nu\in\nan$ and $d\in\nan$. Throughout this article, we fix these numbers.
We denote the algebra of $d\times d$ matrices by $\Mat_{d}$.

We denote the set of all finite subsets in ${\bbZ}^\nu$ by ${\mathfrak S}_{\bbZ^\nu}$.
For each $X\in {\mathfrak S}_{\bbZ^\nu}$, $\diam{X}$ denotes the diameter of $X$.
For $X,Y\subset \bbZ^\nu$, we denote by $d(X,Y)$ the distance between them.
The number of elements in a finite set $\Lambda\subset {\bbZ^\nu}$ is denoted by
$|\Lambda|$. For each $n\in\bbN$, we denote $[-n,n]^\nu\cap \bbZ^\nu$ by $\Lambda_n$.
The complement of $\Lambda\subset\bbZ^\nu$ in $\bbZ^\nu$ is denoted by $\Lambda^c$.

For each $z\in\bbZ^\nu$,  let $\caA_{\{z\}}$ be an isomorphic copy of $\Mat_{d}$, and for any finite subset $\Lambda\subset\bbZ^\nu$, let $\caA_{\Lambda} = \otimes_{z\in\Lambda}\caA_{\{z\}}$, which is the local algebra of observables in $\Lambda$. 
For finite $\Lambda$, the algebra $\caA_{\Lambda} $ can be regarded as the set of all bounded operators acting on
the Hilbert space $\otimes_{z\in\Lambda}{\bbC}^{d}$.
We use this identification freely.
If $\Lambda_1\subset\Lambda_2$, the algebra $\caA_{\Lambda_1}$ is naturally embedded in $\caA_{\Lambda_2}$ by tensoring its elements with the identity. 
The algebra $\caA$, representing the quantum spin system on $\bbZ^\nu$
is given as the inductive limit of the algebras $\caA_{\Lambda}$ with $\Lambda\in{\mathfrak S}_{\bbZ^\nu}$. 
Note that $\caA_{\Lambda}$ for $\Lambda\in {\mathfrak S}_{\bbZ^\nu}$
can be regarded naturally as a subalgebra of
$\caA$.
We denote the set of local observables by $\caA_{\rm loc}=\bigcup_{\Lambda\in{\mathfrak S}_{\bbZ^\nu}}\caA_{\Lambda}
$.

A uniformly bounded interaction on $\caA$ 
is a map $\Psi: {\mathfrak S}_{\bbZ^{\nu}}\to \caA_{\rm loc}$ such that
\begin{align}
\Psi(X)=\Psi(X)^*\in \caA_{X},\quad X\in {\mathfrak S}_{\bbZ^{\nu}},
\end{align}
and 
\begin{align}
\sup_{X\in {\mathfrak S}_{\bbZ^{\nu}}}\lV \Psi(X)\rV<\infty.
\end{align}
It is of finite range with interaction length less than or equal to $R\in\nan$ if 
$\Psi(X)=0$ for any $X\in {\mathfrak S}_{\bbZ^{\nu}}$
whose diameter is larger than $R$.
We denote by $\Psi_n$ for each $n\in\nan$
the interaction given by
\begin{align}
\Psi_n(X):=\left\{
\begin{gathered}
\Psi(X),\quad \text{if}\quad X\subset \Lambda_n,\\
0,\quad \text{otherwise}.
\end{gathered}
\right.
\end{align}

For a uniformly bounded and finite range interaction $\Psi$ and $\Lambda\in {\mathfrak S}_{\bbZ^{\nu}}$
define the local Hamiltonian
\begin{align}
\lmk H_\Psi\rmk_\Lambda
:=\sum_{X\subset\Lambda} \Psi(X),
\end{align}
and denote the dynamics
\begin{align}
\tau_{\Psi,\Lambda}^t (A):=e^{it\lmk H_\Psi\rmk_\Lambda}Ae^{-it\lmk H_\Psi\rmk_\Lambda},
\quad t\in \bbR,\quad A\in\caA.
\end{align}
By the uniform boundedness and finite rangeness of $\Psi$, 
 for each $A\in\caA$, the following limit exists: 
\begin{align}
\lim_{\Lambda\to\bbZ^{\nu}} \tau_{\Psi,\Lambda}^t\lmk
A\rmk=:
\tau_{\Psi}^t\lmk A\rmk,\quad t\in\bbR,
\end{align}
and defines the dynamics $\tau_{\Psi}$ on $\caA$.
Note that $\tau_{\Psi_n}=\tau_{\Psi, \Lambda_n}$.
We denote by $\delta_\Psi$ the generator of $\tau_{\Psi}$.

For a uniformly bounded and finite range interaction $\Psi$,
a state $\varphi$ on $\caA$ is called a \mbox{$\tau_{\Psi}$-ground} state
if the inequality
$
-i\,\varphi(A^*{\delta_{\Psi}}(A))\ge 0
$
holds
for any element $A$ in the domain $\caD({\delta_{\Psi}})$ of ${\delta_\Psi}$.
Let $\varphi$ be a $\tau_\Psi$-ground state, with the GNS triple $(\caH_\varphi,\pi_\varphi,\Omega_\varphi)$.
Then there exists a unique positive operator $H_{\varphi,\Psi}$ on $\caH_\varphi$ such that
$e^{itH_{\varphi,\Psi}}\pi_\varphi(A)\Omega_\varphi=\pi_\varphi(\tau^t_\Psi(A))\Omega_\varphi$,
for all $A\in\caA$ and $t\in\mathbb R$.
We call this $H_{\varphi,\Psi}$ the bulk Hamiltonian associated with $\varphi$.
Note that $\Omega_\varphi$ is an eigenvector of $H_{\varphi,\Psi}$ with eigenvalue $0$. See \cite{BR2} for the general theory.

Let $\bbE_{N}:\caA\to \caA_{\Lambda_N}$ be the conditional expectation with respect to the trace state.
Let us consider the following subset of $\caA$. (See \cite{bdn} and \cite{ma} for analogous definitions.)
\begin{defn}
Let $f:(0,\infty)\to (0,\infty)$ be a continuous decreasing function 
with $\lim_{t\to\infty}f(t)=0$.
For each $A\in\caA$, let
\begin{align}
\lV A\rV_f:=\lV A\rV
+ \sup_{N\in \nan}\lmk\frac{\lV
A-\bbE_{N}(A)
\rV}
{f(N)}
\rmk.
\end{align}
We denote by $\caD_f$ the set of all $A\in\caA$ such that
$\lV A\rV_f<\infty$.
\end{defn}
Properties of $\caD_f$ are collected in Appendix \ref{dfsec}.
The set $\caD_f$ is a $*$-algebra which is a Banach space with respect to 
the norm $\lV\cdot\rV_f$ (see Lemma \ref{alg}).

\begin{assum}\label{assump}
Let  $\Phi (\cdot~ ; s) : \mathfrak{S}_{\mathbb{Z}^\nu } \to \mathcal{A}_{\rm loc}$ be a family of uniformly bounded, finite range interactions parameterized by $s\in [0,1]$. We assume the following:
\begin{description}
\item[(i)]
For each $X\in{\mathfrak S}_{\bbZ^\nu}$, the map
$[0,1]\ni s\to \Phi(X;s)\in\caA_{X}$ is continuous and piecewise $C^1$.
We denote by $\dot{\Phi}(X;s)$ 
the corresponding derivatives.
The interaction obtained by differentiation is denoted by $\dot\Phi(s)$, for each $s\in[0,1]$.
\item[(ii)]
There is a number $R\in\nan$
such that $X \in {\mathfrak S}_{\bbZ^\nu}$ and $\diam{X}\ge R$ imply $\Phi(X;s)=0$, for all $s\in[0,1]$.
\item[(iii)] Interactions are bounded as follows
\begin{align}
\sup_{s\in[0,1]}\sup_{X\in {\mathfrak S}_{\bbZ^\nu}}
\lmk
\lV
\Phi\lmk X;s\rmk
\rV+|X|\lV
\dot{\Phi} \lmk X;s\rmk
\rV
\rmk<\infty.
\end{align}
\item[(iv)]
Setting
\begin{align}
b(\varepsilon):=\sup_{Z\in{\mathfrak S}_{\bbZ^\nu}}
\sup_{s,s_0 \in[0,1],0<| s-s_0|<\varepsilon}
\lV
\frac{\Phi(Z;s)-\Phi(Z;s_0)}{s-s_0}-\dot{\Phi}(Z;s_0)
\rV
\end{align}
for each $\varepsilon>0$, we have
$\lim_{\varepsilon\to 0} b(\varepsilon)=0$.

\item[(v)] For each $s\in[0,1]$, there exists a unique $\tau_{\Phi(s)}$-ground state
$\varphi_s$. 
\item[(vi)] 
There exists a $\gamma>0$ such that
$\sigma(H_{\varphi_s,\Phi(s)})\setminus\{0\}\subset [2\gamma,\infty)$ for
all $s\in[0,1]$, where  $\sigma(H_{\varphi_s,\Phi(s)})$ is the spectrum of $H_{\varphi_s,\Phi(s)}$.
\item[(vii)]
There exists $0<\beta<1$ satisfying the following:
Set $\zeta(t):=e^{-t^{ \beta}}$.
Then for each $A\in D_\zeta$, 
$\varphi_s(A)$ is differentiable with respect to $s$, and there is a constant
$C_\zeta$ such that:
\begin{align}\label{dcon}
\lv
\dot{\varphi_s}(A)
\rv
\le C_\zeta\lV A\rV_\zeta,
\end{align}
for any $A\in D_\zeta$.

\end{description}
\end{assum}
The main theorem of this paper is that
under the Assumption \ref{assump}, there is a strongly continuous path
 of automorphisms $[0,1]\ni s\mapsto \alpha_s$ such that
 $
\varphi_s=\varphi_0\circ\alpha_s,\quad s\in[0,1].
$

In fact, this $\alpha_s$ is the same one as in \cite{BMNS} and \cite{NSY}, which is given through
some differential equation.
Let us recall it.

We use the function $\omega_1$ introduced in \cite{NSY}.
Set
\begin{align}
a_n:=\frac{a_1}{n\ln (n)^2},\quad n\ge 2,
\end{align}
and choose $a_1$ so that $\sum_{n=1}^\infty a_n=\frac 12$.
Let $\omega_1(t)\in L^1(\bbR)$ be the function on $\bbR$ defined by
\begin{align}
\omega_1(t):=
\left\{
\begin{gathered}
c,\quad t=0,\\
c\prod_{n=1}^{\infty} \lmk
\frac{\sin(a_n t)}{a_n t}
\rmk^2,\quad t\neq 0
\end{gathered}
\right.
\end{align}
with normalization factor $c>0$ such that 
\begin{align}
\int dt\omega_1(t) =1.
\end{align}
As shown in \cite{BMNS} and \cite{NSY}, $\omega_1$ is indeed an even nonnegative $L^1$-function and
\begin{align}
&\omega_1(t)
\le
c_{1}\frac{t}{\ln (t)^2}e^{-\frac{\eta  t}{\ln (t)^2}},\quad t>e,
\label{omegabound}\\
&
W_1(x):=\int_x^\infty dt
\omega_1(t) 
\le
\left\{ 
\begin{gathered}
c_1\lmk \frac{x}{\ln (x)^2}\rmk^2 e^{-\frac{\eta x}{\ln (x)^2}},\quad x>e^9,\\
1,\quad x\le e^9
\end{gathered}
\right.\label{Omegabound}
\end{align}
for constants $\eta = 2 a_1 \in (\frac{2}{7}, 1)$ and $c_{1} = (27/14) ce^4$.
We set $\omega_\gamma(t):=\gamma \omega_1(\gamma t)$, where $\gamma>0$ is from Assumption \ref{assump}, and 
$W_\gamma(x):=W_{1}(\gamma x)$, for $x\in \bbR_+$.
The function $\omega_\gamma$ is an even nonnegative $L^1$-function
with \begin{align}\label{omeganom}
\int dt\omega_\gamma(t) =1.
\end{align}
We also have
\begin{align}\label{Wint}
W_\gamma(x)=\int_{x}^\infty dt \omega_\gamma(t),\quad x\in \bbR_{+}.
\end{align}
Furthermore, the Fourier transform of $\omega_\gamma$ is supported in the interval $[-\gamma ,\gamma ]$.(See \cite{NSY}.)

For each $\Lambda\in{\mathfrak S}_{\bbZ^\nu}$,
let $U_\Lambda$ be the solution of the differential equation
\begin{align}\label{udf}
-i\frac{d}{ds} U_{\Lambda}(s)=D_{\Lambda}(s) U_{\Lambda}(s),\quad U_\Lambda(0)=\unit.
\end{align}
Here, $D_{\Lambda}(s)$ is defined by
\begin{align}\label{dkdef}
D_{\Lambda}(s):=
\int_{-\infty}^\infty dt\;
\omega_\gamma(t)\;\;
\int_0^t du \tau^u_{\Phi(s),\Lambda} 
\lmk \frac{d}{ds}\lmk H_{\Phi(s)}\rmk_{\Lambda}\rmk,\quad s\in[0,1].
\end{align}
We set 
\begin{align}
\alpha_{s,\Lambda}(A):=U_{\Lambda}(s)^* A U_{\Lambda}(s),
\quad A\in\caA,\quad s\in[0,1].
\end{align}
From \cite{BMNS} and \cite{NSY}, the thermodynamic limit 
\begin{align}\label{alphadef}
\alpha_{s}(A)=\lim_{\Lambda\to\bbZ^{\nu}} \alpha_{s,\Lambda}(A),\quad A\in\caA,\quad s\in[0,1],
\end{align}
exists and defines a strongly continuous path
 of automorphisms $[0,1]\ni s\mapsto \alpha_s$.
We also have the limit of the inverse
\begin{align}
\alpha_{s}^{-1}(A)=\lim_{\Lambda\to\bbZ^{\nu}} \alpha_{s,\Lambda}^{-1}(A),\quad A\in\caA,\quad s\in[0,1].
\end{align}
See \cite{NSY}.
Our main theorem is as follows.
\begin{thm}\label{main}
Under the Assumption \ref{assump}, we have
\[
\varphi_s=\varphi_0\circ\alpha_s,\quad s\in[0,1],
\] 
for $\alpha_{s}$ given in (\ref{alphadef}).
\end{thm}
Our motivation to develop this bulk version of automorphic equivalence was the 
index theorems for SPT-phases \cite{tri} and \cite{reflection}.
In \cite{tri} and \cite{reflection}, the path of interactions was required to
have a uniform spectral gap for corresponding local Hamiltonians.
It is a bit unpleasant that we have to ask for the existence of the gap for local Hamiltonians
while what we really would like to investigate is the bulk.
From our Theorem \ref{main}, combined with Theorem 2.6, and the proof of Proposition 3.5 of \cite{tri},
we obtain the following version of the index theorem for the time reversal symmetry.
\begin{thm}
Let  $\Phi (\cdot~ ; s) : \mathfrak{S}_{\mathbb{Z}^\nu } \to \mathcal{A}_{\rm loc}$
be a path of time-reversal interactions satisfying Assumption \ref{assump}.
Then  $\bbZ_2$-index defined in Definition 3.3 of \cite{tri} is constant along the path.
\end{thm}
From our Theorem \ref{main}, combined with Theorem 2.9 of \cite{reflection}, and the proof of Proposition 3.5 of \cite{tri},
we obtain the following version of the index theorem for the reflection symmetry.
\begin{thm}
Let  $\Phi (\cdot~ ; s) : \mathfrak{S}_{\mathbb{Z}^\nu } \to \mathcal{A}_{\rm loc}$
be a path of reflection invariant interactions satisfying Assumption \ref{assump}.
Then  $\bbZ_2$-index defined in Definition 3.3 of \cite{reflection} is constant along the path.
\end{thm}

The rest of the paper is devoted to the proof of Theorem \ref{main}.

\section{Proof of the Theorem \ref{main}}\label{thmprf}
Throughout this Section, we will always assume
Assumption \ref{assump}.
For $s\in[0,1]$ and $A\in\caA$, we set
\begin{align}
I_s(A):=
\int dt \;\omega_\gamma(t) \tau_{\Phi(s)}^t(A) .
\end{align}
The integral can be understood as a Bochner integral of $(\caA, \lV\cdot\rV)$.

We need the following 
Lemma for the proof.
\begin{lem}\label{pre}
Fix $0<\beta=\beta_5<\beta_4<\beta_3<\beta_2<\beta_1<1$ and set
$f(t):=t^{-1}\exp(-t^{\beta_1})$, 
$f_0(t):=\exp(-t^{\beta_1})$,
$f_1(t):=\exp(-t^{\beta_2})$, $f_2(t):=t^{-2(\nu+2)}\exp(-t^{\beta_3})$, $g(t):=\exp(-t^{\beta_4})$, $\zeta(t):=\exp(-t^{\beta_5})$.
(Here $\beta$ is the one in (vii) of Assumption \ref{assump}.)
Then we have the following.
\begin{enumerate}
\item For any $s\in[0,1]$, we have
\[
\alpha_s^{-1}(\caA_{\rm loc})\subset \caD_{f}
\subset \caD_{f_0}\subset \caD_{f_1}\subset\caD_{f_2}\subset \caD_g\subset \caD_\zeta.
\]
\item 
We have $\tau_{\Phi(s)}^t\lmk\caD_{f}\rmk\subset\caD_{{f_1}}$ and
there is a non-negative non-decreasing function on $\bbR_{\ge 0}$, $b_{f,{f_1}}(t)$ such that
\begin{align}\label{b11}
\int\; dt\; \omega_\gamma(t)\; |t|\cdot b_{f,{f_1}}(|t|)<\infty,
\end{align}
\begin{align}\label{b1t}
&
\sup_{s\in[0,1]}\lV
\tau_{\Phi(s)}^t\lmk
A
\rmk
\rV_{{f_1}}
\le
b_{f,{f_1}}(|t|)\lV A\rV_f,\quad A\in\caD_f.
\end{align}
\item
We have $\caD_{\zeta}\subset D(\delta_{\Phi(s)})\cap D(\delta_{\dot\Phi(s)})$ for any $s\in[0,1]$.
\item
There is a constant $C^{(1)}_{f_2,\zeta}>0$
such that 
\begin{align}
&\sup_{s\in[0,1]}\lV
\delta_{\Phi(s)}\lmk A\rmk
\rV_{\zeta},\;\sup_{N\in\nan}\sup_{s\in[0,1]}\lV
\delta_{\Phi_N(s)}\lmk A\rmk
\rV_{\zeta}
\le
C^{(1)}_{f_2,\zeta}\lV A\rV_{f_2}\label{deltap}\\
&\sup_{s\in[0,1]}\lV
\delta_{\dot{\Phi}(s)}\lmk A\rmk
\rV_{\zeta},\;
\sup_{N\in\nan}\sup_{s\in[0,1]}\lV
\delta_{\dot{\Phi}_N(s)}\lmk A\rmk
\rV_{\zeta}
\le
C^{(1)}_{f_2,\zeta}\lV A\rV_{f_2}\label{deldadp}\\
&\sup_{s,s_0 \in[0,1],0<| s-s_0|\le \varepsilon}\lV
\delta_{\frac{\Phi(s)-\Phi(s_0)}{s-s_0}-\dot{\Phi}(s_0)}
\lmk A\rmk
\rV_{\zeta},\;
\sup_{N\in\nan}\sup_{s,s_0 \in[0,1],0<| s-s_0|\le \varepsilon}\lV
\delta_{\frac{\Phi_N(s)-\Phi_N(s_0)}{s-s_0}-\dot{\Phi}_N(s_0)}
\lmk A\rmk
\rV_{\zeta}\nonumber\\
&\le
b(\varepsilon)C^{(1)}_{f_2,\zeta}\lV A\rV_{f_2}\label{dddb}
\end{align}
for all $A\in\caD_{f_2}$.
(Here the meaning of the inequality is that each term on the left hand side
is bounded by the right hand side. We use this notation
throughout this article.)
In particular, $\delta_{\Phi(s)}(\caD_{f_2})\subset \caD_\zeta$,
for any $s\in[0,1]$.
(Recall $b(\varepsilon)$ in Assumption \ref{assump} {\it (iv)}.)
\item
For any $A\in \caD_f$, and $(s',u',s'',s''')\in [0,1]\times\bbR\times [0,1]\times[0,1]$,
we have
$\tau_{\Phi(s'')}^{-u'}\circ \alpha_{s'''}^{-1}(A)\in\caD_{f_{2}}\subset \caD_{\zeta}\subset  D(\delta_{\Phi(s')})\cap D(\delta_{\dot\Phi(s')})$
and $\delta_{{\Phi}(s')}\circ
\tau_{\Phi(s'')}^{-u'}\circ \alpha_{s'''}^{-1}(A),
\delta_{\dot{\Phi}(s')}\circ
\tau_{\Phi(s'')}^{-u'}\circ \alpha_{s'''}^{-1}(A)\in\caD_{\zeta}$.
For any  compact intervals $[a,b]$, $[c,d]$ of $\bbR$
and $A\in\caD_f$, the maps:
\begin{align}
[a,b]\times [0,1]\times [0,1]\times[c,d]\times [0,1]\times[0,1]
\ni (u,s,s',u',s'',s''')\mapsto
 \tau_{\Phi(s)}^{u}\circ\delta_{{\Phi}(s')}\circ
\tau_{\Phi(s'')}^{-u'}\circ \alpha_{s'''}^{-1}(A)
\in\caA,
\end{align}
and 
\begin{align}
[a,b]\times [0,1]\times [0,1]\times[c,d]\times [0,1]\times[0,1]
\ni (u,s,s',u',s'',s''')\mapsto
 \tau_{\Phi(s)}^{u}\circ\delta_{\dot{\Phi}(s')}\circ
\tau_{\Phi(s'')}^{-u'}\circ \alpha_{s'''}^{-1}(A)
\in\caA
\end{align}
are uniformly continuous with respect to $\lV\cdot\rV$, and maps
\begin{align}\label{392a}
 [0,1]\times[c,d]\times [0,1]\times[0,1]
\ni (s',u',s'',s''')\mapsto
\delta_{{\Phi}(s')}\circ
\tau_{\Phi(s'')}^{-u'}\circ \alpha_{s'''}^{-1}(A)
\in\caD_{\zeta}
\end{align}
\begin{align}\label{392}
 [0,1]\times[c,d]\times [0,1]\times[0,1]
\ni (s',u',s'',s''')\mapsto
\delta_{\dot{\Phi}(s')}\circ
\tau_{\Phi(s'')}^{-u'}\circ \alpha_{s'''}^{-1}(A)
\in\caD_{\zeta}
\end{align}
are uniformly continuous with respect to $\lV\cdot\rV_\zeta$.
\item
For any $A\in \caD_f$,
$\alpha_s^{-1}(A)$ is differentiable with respect to
$\lV\cdot\rV$ and
\begin{align}
\frac{d}{ds}\alpha_s^{-1}(A)
=\int dt \omega_\gamma(t) \int_0^t du
\tau_{\Phi(s)}^u\circ\delta_{\dot{\Phi}(s)}\lmk
\tau_{\Phi(s)}^{-u}\lmk \alpha_{s}^{-1}(A)\rmk
\rmk
\end{align}
The right hand side can be understood as a Bochner integral of ($\caA$, $\lV \cdot \rV$).
\item
For any $A\in \caD_f$, the integral 
\begin{align}\label{vswe}
\int dt \omega_\gamma(t) \int_0^t du
\tau_{\Phi(s)}^u\circ\delta_{\dot{\Phi}(s)}\lmk
\tau_{\Phi(s)}^{-u}(A)\rmk
\end{align}
\begin{align}
\int dt \;\omega_\gamma(t)\int_0^t du\;
\tau_{\Phi(s)}^{t-u}\circ \lmk \delta_{\dot\Phi(s)}\rmk\circ
\tau_{\Phi(s)}^u(A)
\end{align}
are well-defined as 
a Bochner integral with respect to
($\caA$, $\lV \cdot \rV$).

\item
For any $A \in \caD_f$ and $s\in[0,1]$, we have $I_s(A) \in \caD_{f_1}$.
\item
For each $A\in\caA$,
$\bbR\times [0,1]\ni (u,s)\to \tau_{\Phi(s)}^{u}(A)\in \caA$ is 
continuous with respect to the norm $\lV \cdot\rV$.
\item
For any $A\in\caD_f$, the integrals
\begin{align}\label{zetabd}
\int dt \omega_\gamma(t) 
\int_0^t du \delta_{\Phi(s)}\circ\tau_{\Phi(s)}^{u}(A),\quad
\int_0^t du \delta_{\Phi(s)}\circ\tau_{\Phi(s)}^{u}(A),
\end{align}
 are well-defined as Bochner integrals with respect to
$(\caD_{\zeta}, \lV \cdot\rV_\zeta)$.

\end{enumerate}
\end{lem}
The proof of Lemma \ref{pre} is given in Section \ref{lems}.
Throughout Section \ref{thmprf} and Section \ref{if} (but not in Section \ref{lems}),
we fix $0<\beta_5<\beta_4<\beta_3<\beta_2<\beta_1<1$ and set
$f, f_0, f_1,f_2,g, \zeta$, given in Lemma \ref{pre}, and apply Lemma \ref{pre}.

In Section \ref{if}, we prove the following:
\begin{prop}\label{iff}
For any $A\in \mathcal{D}_f$, we have
\begin{align}
\dot{\varphi_s} \lmk
I_s(A)\rmk=0
,\quad s\in[0,1].
\end{align}
\end{prop}
Note that by {\it 8.} of Lemma \ref{pre}, $I_s(A)$ belongs to
$\caD_{f_1}\subset \caD_{\zeta}$, and that 
$\dot{\varphi_s} \lmk
I_s(A)\rmk$ in Proposition \ref{iff} is well-defined by (vii) of Assumption \ref{assump}.

We now prove Theorem \ref{main} using this proposition.
In order to prove the Theorem, it suffices to show
\begin{align}\label{dz}
\frac d{ds} \lmk {\varphi_s}\circ\alpha_s^{-1}(X)\rmk=0,
\end{align}
for any $X\in\caA_{\rm loc}$.
Note that from Assumption \ref{assump} (vii), and {\it 1.} of Lemma \ref{pre},
the function $[0,1]\ni s\to \varphi_s\circ\alpha_{s_0}^{-1}(X)$ is differentiable for any  $X\in\caA_{\rm loc}$ and $s_0\in[0,1]$.
Furthermore, from {\it 6.} of Lemma \ref{pre}, $[0,1]\ni s\mapsto \alpha_{s}^{-1}(X)\in\caA$
is differentiable with respect to the norm for any  $X\in\caA_{\rm loc}\subset \caD_f$.
Therefore, for any $X\in\caA_{\rm loc}$,
$[0,1]\ni s\to {\varphi_s}\circ\alpha_s^{-1}(X)$ is differentiable, the left hand side of
(\ref{dz}) makes sense, and we have
\begin{align}\label{bbbbs}
\frac d{ds} \lmk {\varphi_s}\circ\alpha_s^{-1}(X)\rmk
=\dot{\varphi_s}\circ\alpha_s^{-1}(X)
+\varphi_s\circ \frac{d}{ds}\alpha_s^{-1}(X),\quad
X\in\caA_{\rm loc}.
\end{align}
For the proof of (\ref{dz}), we use the following Lemma.
\begin{lem}\label{nini}
For any $A\in \mathcal{D}_f$, 
\begin{align}\label{aisa}
A-I_s(A)
=-\int dt \omega_\gamma(t) 
\int_0^t du \delta_{\Phi(s)}\circ\tau_{\Phi(s)}^{u}(A).
\end{align}
%
%
The integrand of the right hand side is continuous with respect to $\lV \cdot\rV_\zeta$
and the integral can be understood as the Bochner integral of
$(\caD_{\zeta}, \lV \cdot\rV_\zeta)$.
\end{lem}
\begin{proof}

The latter part is {\it 5., 10.} of lemma \ref{pre}. To show (\ref{aisa}), 
recall the Duhamel formula
\begin{align}\label{duhamel1}
A-\tau^t_{\Phi(s)}(A)
=\int_0^t du\;
\lmk -\delta_{\Phi(s)}\rmk\circ
\tau_{\Phi(s)}^u(A),\quad A\in\caD_f.
\end{align}
Here we used the fact that 
$\tau_{\Phi(s_0)}^u\lmk \caD_f\rmk\subset
\caD_{f_1}\subset \caD_\zeta\subset D\lmk \delta_{\Phi(s)}\rmk$, which follows from {\it 2.,1,.3.}
of Lemma \ref{pre}.

We multiply (\ref{duhamel1}) by $\omega_\gamma(t)$ and integrate over $t\in\bbR$.
Then recalling (\ref{omeganom}), we obtain 
\begin{align}
	\begin{split}
&A-I_{s}(A)=
\int dt \;\omega_\gamma(t)A-\int dt \;\omega_\gamma(t)\tau^t_{\Phi(s)}(A)\\
&=\int dt \;\omega_\gamma(t)\int_0^t du\;
\lmk -\delta_{\Phi(s)}\rmk\circ
\tau_{\Phi(s)}^u(A),\quad A\in\caD_f.
	\end{split}
\end{align}
\end{proof}
In order to show (\ref{dz}), we need to know $\dot{\varphi_s}$ on $\mathcal{D}_f$.
From Proposition \ref{iff} and Lemma \ref{nini},
for any $A\in \mathcal{D}_f$,
we have
\begin{align}\label{san}
\lmk \dot{\varphi_s}\rmk(A)
=\lmk \dot{\varphi_s}\rmk(A)-\lmk
\dot{\varphi_s}\rmk\lmk
I_s(A)
\rmk
=-\int dt \omega_\gamma(t) 
\int_0^t du \dot{\varphi_s}\lmk
\delta_{\Phi(s)}\circ\tau_{\Phi(s)}^{u}(A)\rmk.
\end{align}
Here we used the Bochner integrability of the right hand side of (\ref{aisa}) with respect to
 $\lV \cdot\rV_\zeta$, and the continuity of $\dot\varphi_s$
(\ref{dcon}) with respect to $\lV \cdot\rV_\zeta$.

As $\varphi_s$ is the $\tau_{\Phi(s)}$-ground state, we have
\begin{align}
\varphi_s\circ\delta_{\Phi(s)}(B)=0,\quad B\in \mathcal{D}_{f_1},\quad s\in[0,1].
\end{align}
(Recall that $\caD_{f_1}\subset \caD_{\zeta}\subset D(\delta_{\Phi(s)})$, from {\it 1., 3.} of
Lemma \ref{pre}.)
Differentiating this by $s$, 
 we obtain
 \begin{align}\label{52}
\dot{\varphi_s}\circ\delta_{\Phi(s)}(B)+\varphi_s\circ\delta_{\dot{\Phi}(s)}(B)
=0,\quad B\in \mathcal{D}_{f_1},\quad s\in[0,1].
\end{align}
More precisely, note that 
\begin{align}
\delta_{\Phi(s)}\lmk\caD_{f_1}\rmk
\subset \delta_{\Phi(s)}\lmk\caD_{f_2}\rmk
\subset \caD_\zeta,\quad s\in[0,1],
\end{align}
by Lemma \ref{pre}, {\it 1., 4.}.
Therefore, for $B\in \caD_{f_1}$, we have $\delta_{\Phi(s)}(B)\in\caD_\zeta$, 
$s\in[0,1]$, and for any $s,s_0\in[0,1]$ with $s\neq s_0$,
 we have
\begin{align}
	\begin{split}
&\lv-\lmk
\dot{\varphi_{s_0}}\circ\delta_{\Phi(s_0)}(B)+\varphi_{s_0}\circ\delta_{\dot{\Phi}(s_0)}(B)
\rmk\rv\\
&=
\lv\frac{\varphi_s\circ\delta_{\Phi(s)}(B)-\varphi_{s_0}\circ\delta_{\Phi(s_0)}(B)}{s-s_0}
-\lmk
\dot{\varphi_{s_0}}\circ\delta_{\Phi(s_0)}(B)+\varphi_{s_0}\circ\delta_{\dot{\Phi}(s_0)}(B)
\rmk\rv\\
&\le
\lv\varphi_s\lmk
\frac{
\delta_{\Phi(s)}(B)-\delta_{\Phi(s_0)}(B)}{s-s_0}- \delta_{\dot{\Phi}(s_0)}(B)\rmk\rv
+\lv
\frac{\varphi_s\circ\delta_{\Phi(s_0)}(B)-\varphi_{s_0}\circ\delta_{\Phi(s_0)}(B)}{s-s_0}
-\lmk
\dot{\varphi_{s_0}}\circ\delta_{\Phi(s_0)}(B)
\rmk\rv\\
&+\lv\lmk \varphi_s-\varphi_{s_0}\rmk\lmk
 \delta_{\dot{\Phi}(s_0)}(B)
\rmk\rv.
	\end{split}
\end{align}
As $\delta_{\Phi(s_0)}(B)\in\caD_\zeta$, the second and the third terms of the last line 
converge to $0$ as $s\to s_0$.
The first term of the last line can be bounded as
\begin{align}
	\begin{split}
&\lv\varphi_s\lmk
\frac{
\delta_{\Phi(s)}(B)-\delta_{\Phi(s_0)}(B)}{s-s_0}- \delta_{\dot{\Phi}(s_0)}(B)\rmk\rv
\le
\lV
\frac{\delta_{\Phi(s)}(B)-\delta_{\Phi(s_0)}(B)}{s-s_0}- \delta_{\dot{\Phi}(s_0)}(B)
\rV
\\
&\le
b(|s-s_0|)C_{f_2,\zeta}^{(1)}\lV B\rV_{f_2}
\to 0,\quad s\to s_0,
	\end{split}
\end{align}
and goes to $0$ as $s\to s_0$.
Here, in the last line, we used {\it 4.} of Lemma \ref{pre} and
recalled $\caD_{f_1}\subset\caD_{f_2}$, from {\it 1.} of Lemma \ref{pre}, and
(iv) of Assumption \ref{assump}.
Hence we obtain (\ref{52}).

From this and (\ref{san}), for $A\in\caD_{f}$, recalling 
$\tau_{\Phi(s)}^{u}(A)\in\caD_{f_1}$ by 
{\it 2.} of Lemma \ref{pre},
we have
\begin{align}\label{yon}
\lmk \dot{\varphi_s}\rmk(A)
=\int dt \omega_\gamma(t) 
\int_0^t du 
\varphi_s\circ\delta_{\dot{\Phi}(s)}\lmk\tau_{\Phi(s)}^{u}(A)\rmk.
\end{align}
For any $X\in\caA_{\rm loc}$, recall that 
$\alpha_s^{-1}(X)\in\alpha_s^{-1}\lmk\caA_{\rm loc}\rmk\subset\caD_f\subset \caD_{\zeta}$
by {\it 1.} of Lemma \ref{pre}.
From (\ref{bbbbs}), (\ref{yon}) and {\it 6.} of Lemma \ref{pre}, we have 
\begin{align}
&\frac d{ds} \lmk {\varphi_s}\circ\alpha_s^{-1}(X)\rmk
=\dot{\varphi_s}\circ\alpha_s^{-1}(X)
+\varphi_s\circ \frac{d}{ds}\alpha_s^{-1}(X)\nonumber\\
&=\int dt \omega_\gamma(t) 
\int_0^t du 
\varphi_s\circ\delta_{\dot{\Phi}(s)}\lmk\tau_{\Phi(s)}^{u}\circ\alpha_s^{-1}(X)\rmk
+\int dt \omega_\gamma(t) \int_0^t du
\varphi_s\lmk\tau_{\Phi(s)}^u\circ\delta_{\dot{\Phi}(s)}\lmk
\tau_{\Phi(s)}^{-u}\lmk \alpha_{s}^{-1}(X)\rmk
\rmk\rmk=0
\end{align}
Here we used the fact that $\omega_\gamma$ is an even function, and that 
$\varphi_s$ is $\tau_{\Phi(s)}$-invariant because it is the $\tau_{\Phi(s)}$-ground state.

Hence we have proven the Theorem \ref{main}.

\section{Proof of Proposition \ref{iff}}\label{if}
Throughout this Section, we keep
Assumption \ref{assump}. We also continue to use the
same $0<\beta=\beta_5<\beta_4<\beta_3<\beta_2<\beta_1<1$ and set
$f, f_0, f_1,f_2,g, \zeta$, as given in Lemma \ref{pre}.

Let $(\caH_s,\pi_s,\Omega_s)$ be the GNS triple of $\varphi_s$.
Let $H_s:=H_{\varphi_{s},\Phi(s)}$ be the associated bulk Hamiltonian.
The key property of $I_s$ we use is the following.
\begin{lem}\label{key}
For any $A\in \caA$, we have
\begin{align}
\pi_s\lmk I_s\lmk A\rmk\rmk\Omega_s
=\varphi_s(A) \Omega_s.
\end{align}
\end{lem}
\begin{proof}
As the Fourier transform $\hat \omega_\gamma$
of $\omega_\gamma$ has support in $[-\gamma,\gamma]$, (v) and (vi) of Assumption \ref{assump} and (\ref{omeganom}) implies:
\begin{align}\label{fourie}
\hat \omega_\gamma\lmk H_s\rmk=\frac{1}{\sqrt{2\pi}}
\ket{\Omega_s}\bra{\Omega_s}.
\end{align}
From the definition of $I_s$, substituting (\ref{fourie}), we have
\begin{align}
	\begin{split}
&\pi_s\lmk I_s\lmk A\rmk\rmk\Omega_s
=\int dt \;\omega_\gamma(t) 
\pi_s\lmk \tau_{\Phi(s)}^t(A)\rmk
\Omega_s\\
&=\int dt \;\omega_\gamma(t) 
 e^{itH_{s}}\pi_s(A)
\Omega_s
={\sqrt{2\pi}}\hat{\omega}_\gamma(H_s)\pi_s(A)\Omega_s
=\varphi_s(A)\Omega_s.
	\end{split}
\end{align}

\end{proof}

From this, we immediately obtain the following decoupling.
\begin{lem}\label{decomi}
For any $A,B\in\caA$ and $s\in[0,1]$,
we have
\begin{align}\label{go}
\varphi_s\lmk
B^*I_s(A)
\rmk
=\varphi_s(B^*)\varphi_s(A).
\end{align}
\end{lem}

\begin{lem}\label{lemv}
For each $s\in[0,1]$ and $A\in\caD_f$, the integrand of
\begin{align}\label{vs0}
V_{s}(A):=
\int dt \;\omega_\gamma(t)\int_0^t du\;
\tau_{\Phi(s)}^{t-u}\circ \lmk \delta_{\dot\Phi(s)}\rmk\circ
\tau_{\Phi(s)}^u(A),
\end{align}
 is continuous 
 and the integral can be understood as a Bochner integral in Banach space
 $(\caA,\lV\cdot \rV)$.
For any $A\in D_f$,
$[0,1]\ni s\to I_s(A)\in\caA$ is differentiable with respect to $\lV\cdot\rV$ and 
\begin{align}
\frac{d}{ds}I_s(A)=V_{s}(A).
\end{align}
%
\end{lem}
\begin{proof}
Let $A\in\caD_f$.
That the integrand of (\ref{vs0}) is continuous 
 and the integral can be understood as a Bochner integral in Banach space
$(\caA,\lV\cdot \rV)$,
follow from  {\it 5.} and {\it 7.,} of Lemma \ref{pre}, respectively.

Next, recall the Duhamel formula
\begin{align}\label{duhamel}
\tau^t_{\Phi(s)}(A)-\tau^t_{\Phi(s_0)}(A)
=\int_0^t du\;
\tau_{\Phi(s)}^{t-u}\circ \lmk \delta_{\Phi(s)}-\delta_{\Phi(s_0)}\rmk\circ
\tau_{\Phi(s_0)}^u(A),\quad A\in \caD_f.
\end{align}
Here we used the fact that 
$\tau_{\Phi(s_0)}^u\lmk \caD_f\rmk\subset \caD_\zeta\subset D\lmk \delta_{\Phi(s)}\rmk$, which follows from {\it 2., 1., 3.,} of
Lemma \ref{pre}.
By {\it 5.} of Lemma \ref{pre}, 
 the integrand on the right hand side is continuous 
 and the integral can be understood as a Bochner integral in Banach space
 $(\caA,\lV\cdot \rV)$.
%
 
We multiply (\ref{duhamel}) by $\omega_\gamma(t)$ and integrate over $t\in\bbR$.
Then we obtain 
\begin{align}
	\begin{split}
&I_s(A)-I_{s_0}(A)=
\int dt \;\omega_\gamma(t)\tau^t_{\Phi(s)}(A)-\int dt \;\omega_\gamma(t)\tau^t_{\Phi(s_0)}(A)\\
&=\int dt \;\omega_\gamma(t)\int_0^t du\;
\tau_{\Phi(s)}^{t-u}\circ \lmk \delta_{\Phi(s)}-\delta_{\Phi(s_0)}\rmk\circ
\tau_{\Phi(s_0)}^u(A),\quad A\in\caD_f.
	\end{split}
\end{align}
By {\it 5.} of Lemma \ref{pre}, all the integrands are 
 continuous 
 and the integral can be understood as a Bochner integral in Banach space
$(\caA,\lV\cdot \rV)$.
For any $A\in \caD_f$,
\begin{align}
	\begin{split}
&\lV\frac{I_s(A)-I_{s_0}(A)}{s-s_0}-V_{s_0}(A)\rV\\
&\le
\int dt \;\omega_\gamma(t)\int_{[0,t]} du\;
\lV
\tau_{\Phi(s)}^{t-u}\circ \lmk \frac{\delta_{\Phi(s)}-\delta_{\Phi(s_0)}}{s-s_0}\rmk\circ
\tau_{\Phi(s_0)}^u(A)
-\tau_{\Phi(s_0)}^{t-u}\circ \lmk \delta_{\dot\Phi(s_0)}\rmk\circ
\tau_{\Phi(s_0)}^u(A)
\rV\\
&\le
\int dt \;\omega_\gamma(t)\int_{[0,t]} du\;
\lmk
\begin{gathered}
\lV
\lmk\tau_{\Phi(s)}^{t-u}-\tau_{\Phi(s_0)}^{t-u}
\rmk\circ
\lmk \delta_{\dot\Phi(s_0)}\rmk\circ
\tau_{\Phi(s_0)}^u(A)
\rV\\
+
\lV
\tau_{\Phi(s)}^{t-u}\circ
 \lmk \frac{\delta_{\Phi(s)}-\delta_{\Phi(s_0)}}{s-s_0}- \lmk \delta_{\dot\Phi(s_0)}\rmk\rmk\circ
\tau_{\Phi(s_0)}^u(A)
\rV
\end{gathered}
\rmk.
	\end{split}
\end{align}
Here and after, $\int_{[0,t]}du$ always indicates
Lebesgue integral (i.e. without sign) over the measurable set $[0,t]$.  
From {\it 9.} of Lemma \ref{pre},
for each $t,u$, we have
\begin{align}
\lim_{s\to s_0}\lV
\lmk\tau_{\Phi(s)}^{t-u}-\tau_{\Phi(s_0)}^{t-u}
\rmk\circ
\lmk \delta_{\dot\Phi(s_0)}\rmk\circ
\tau_{\Phi(s_0)}^u(A)
\rV=0,\quad A\in \caD_f.
\end{align}
By {\it 4.} of Lemma \ref{pre}, for each $t,u$, we have
\begin{align}
\lim_{s\to s_0}\lV
\tau_{\Phi(s)}^{t-u}\circ
 \lmk \frac{\delta_{\Phi(s)}-\delta_{\Phi(s_0)}}{s-s_0}- \lmk \delta_{\dot\Phi(s_0)}\rmk\rmk\circ
\tau_{\Phi(s_0)}^u(A)
\rV
\le \limsup_{s\to s_0}b(|s-s_0|)C_{f_2,\zeta}^{(1)}\lV \tau_{\Phi(s_0)}^u(A)\rV_{f_2}=0,\quad
A\in \caD_f.
\end{align}
Here we used $\tau_{\Phi(s_0)}^u(A)\in \caD_{f_{1}}\subset \caD_{f_{2}}$ which follows from 
Lemma \ref{pre}, {\it 1.,2}.
Furthermore, from {\it 2., 4.} of Lemma \ref{pre}, for $ A\in\caD_{f}$,
\begin{align}
&\lV
\lmk\tau_{\Phi(s)}^{t-u}-\tau_{\Phi(s_0)}^{t-u}
\rmk\circ
\lmk \delta_{\dot\Phi(s_0)}\rmk\circ
\tau_{\Phi(s_0)}^u(A)
\rV
\le
2C_{f_2,\zeta}^{(1)}\lV \tau_{\Phi(s_0)}^u(A)\rV_{f_2}\notag \\
&\le
2C_{f_2,\zeta}^{(1)}
\lmk 1+\sup_N\frac{f_1(N)}{f_2(N)}\rmk\lV \tau_{\Phi(s_0)}^u(A)\rV_{f_1}
\le
2C_{f_2,\zeta}^{(1)}b_{f,f_1}(|u|)
\lmk 1+\sup_N\frac{f_1(N)}{f_2(N)}\rmk\lV A\rV_{f}.\label{ttta}
\end{align}
Note that from $0<\beta_3<\beta_2<1$, we have
$\sup_N\frac{f_1(N)}{f_2(N)}<\infty$.
Similarly, from {\it 2., 4.} of Lemma \ref{pre},
\begin{align}
\lV
\tau_{\Phi(s)}^{t-u}\circ
 \lmk \frac{\delta_{\Phi(s)}-\delta_{\Phi(s_0)}}{s-s_0}- \lmk \delta_{\dot\Phi(s_0)}\rmk\rmk\circ
\tau_{\Phi(s_0)}^u(A)
\rV
\le
b(1)C_{f_2,\zeta}^{(1)}b_{f,f_1}(|u|)
\lmk 1+\sup_N\frac{f_1(N)}{f_2(N)}\rmk\lV A\rV_f.
\end{align}
Combining this (\ref{b11}) in {\it 2.} of Lemma \ref{pre},
from Lebesgue's convergence theorem, we obtain
\begin{align}
\lim_{s\to s_0}\lV\frac{I_s(A)-I_{s_0}(A)}{s-s_0}-V_{s_0}(A)\rV=0,\quad
A\in \caD_f.
\end{align}
\end{proof}
\begin{lem}\label{dff}
For any $A,B\in \mathcal{D}_f$ and $s\in[0,1]$, 
$A,B^{*},B^*I_s(A)$ belong to $\caD_\zeta$ and
 we have
\begin{align}\label{ccca}
&\dot{\varphi_s}\lmk B^*I_s(A)\rmk+
\int dt \;\omega_\gamma(t)\int_0^t du 
\varphi_s\lmk
B^* \tau_{{\Phi(s)}}^{t-u}\circ
\delta_{\dot{\Phi}(s)}\circ \tau_{\Phi(s)}^{u}(A)
\rmk\nonumber\\
&=
\dot{\varphi_s}(B^*)\varphi_s(A)
+\varphi_s(B^*)\dot{\varphi_s}(A).
\end{align}
\end{lem}
\begin{proof}
For any $A,B\in\caD_f\subset\caD_\zeta$ and $s_0\in[0,1]$,
$B^*I_{s_0}(A)$ belongs to $\caD_{f_1}\subset\caD_{\zeta}$
because of {\it 1., 8.,} of Lemma \ref{pre} and Lemma \ref{alg}.
Therefore, by (vii) of Assumption \ref{assump}, 
$[0,1]\ni s\mapsto \varphi_s\lmk B^*I_{s_0}(A)\rmk\in \bbC$ is differentiable.
For any
$s,s_0\in[0,1]$ with $s\neq s_0$, we have
\begin{align}\label{rrh}
&\frac{1}{s-s_0}\lmk \varphi_s\lmk B^*I_s(A)\rmk-\varphi_{s_0}\lmk B^*I_{s_0}(A)\rmk\rmk-
\varphi_{s_0}\lmk
B^{*}V_{s_0}(A)
\rmk-\dot{\varphi}_{s_0}\lmk B^*I_{s_0}(A)\rmk\\
&=\varphi_s\lmk B^*\lmk \frac{I_s(A)-I_{s_0}(A)}{s-s_0}-V_{s_0}(A)\rmk\rmk-\dot{\varphi_{s_0}}\lmk B^*I_{s_0}(A)\rmk
+\frac{1}{s-s_0}\lmk \varphi_s-\varphi_{s_0}\rmk\lmk B^*I_{s_0}(A)\rmk\notag \\
&+\lmk
\varphi_{s}-\varphi_{s_0}
\rmk
\lmk B^{*}
V_{s_0}(A)
\rmk. \notag
\end{align}
The right hand side goes to $0$ as $s\to s_0$, because of Lemma \ref{lemv} and the differentiability of 
 $[0,1]\ni s\mapsto \varphi_s\lmk B^*I_{s_0}(A)\rmk\in \bbC$.
 On the other hand, the first part of the left hand side of (\ref{rrh})
 is \begin{align}
 \frac{1}{s-s_0}\lmk \varphi_s\lmk B^*I_s(A)\rmk-\varphi_{s_0}\lmk B^*I_{s_0}(A)\rmk\rmk
 =\frac{1}{s-s_0}\lmk \varphi_s\lmk B^*\rmk \varphi_{s}\lmk A\rmk-
  \varphi_{s_{0}}\lmk B^*\rmk \varphi_{s_{0}}\lmk A\rmk
 \rmk,
 \end{align}
 because of Lemma \ref{decomi}
 and converges to 
 \begin{align}
 \dot{\varphi_{s_{0}}}(B^*)\varphi_{s_{0}}(A)
+\varphi_{s_{0}}(B^*)\dot{\varphi_{s_{0}}}(A),
 \end{align}
 as $s\to s_0$.
 Hence we obtain (\ref{ccca}).
 
\end{proof}

For each $s\in[0,1]$, we introduce the left ideal $\caL_s$ of $\caA$ by
\begin{align}
\caL_s:=\left\{
A\in\caA\mid
\varphi_s(A^*A)=0
\right\}.
\end{align}
\begin{lem}\label{lld}
For any $A\in \mathcal{D}_f$ and $s\in[0,1]$,
$I_s(A)-\varphi_s(A)\unit$ 
belongs to $\caL_s\cap\caL_s^*\cap \mathcal{D}_{f_1}$.
\end{lem}
\begin{proof}
Let $A\in \caD_f$.
Let $(\caH_s,\pi_s,\Omega_s)$ be the GNS triple of
$\varphi_s$.
That $I_s(A)-\varphi_s(A)\unit\in \mathcal{D}_{f_1}$ is Lemma \ref{pre} {\it 8.}.
To show $I_s(A)-\varphi_s(A)\unit\in \caL_s\cap \caL_s^{*}$,
recall Lemma \ref{key}.
From the latter Lemma, we obtain
\begin{align}
\pi_s\lmk
I_s(A)-\varphi_s(A)
\rmk\Omega_s=
\pi_s\lmk
I_s(A^*)-\varphi_s(A^*)
\rmk\Omega_s=0,
\end{align}
which
means $I_s(A)-\varphi_s(A)\unit\in \caL_s\cap \caL_s^*$, because
$I_s(A)^*=I_s(A^*)$.
\end{proof}
\begin{lem}\label{una}
For any $A\in \caL_s\cap \mathcal{D}_{f_1}$, there is a positive sequence
$u_{N,A}\in  \caA_{\Lambda_N}$,  $N\in\nan$ with $\lV u_{N,A}\rV\le 1$ such that
\begin{align}
\lV A(1-u_{N,A})\rV_g\to 0,
\end{align}
and
\begin{align}\label{psz}
\lim_{N\to\infty}\varphi_s(u_{N,A})= 0,
\end{align}
and
\begin{align}\label{psz2}
{\mathrm{dist}}\lmk
u_{N,A},\caL_s
\rmk
:=\inf_{x\in\caL_s}\lV x-u_{N,A}\rV\to 0,\quad N\to\infty.
\end{align}
\end{lem}
\begin{proof}
Choose $\beta_4<\beta'<\beta_2$ and set $h(t):=e^{t^{\beta'}}$.
Then we have
\begin{align}\label{hasy}
\lim_{N\to \infty}\frac{1}{g(N)\sqrt{h(N)}}=0,\quad
\lim_{N\to\infty} h(N) {f_1(N)}=0.
\end{align}
Let $A\in \caL_s\cap \mathcal{D}_{f_1}$.
Set 
\begin{align}
u_{N,A}:=
\lmk
1+h(N)\bbE_N(A^*A)
\rmk^{-1}
h(N)\bbE_N(A^*A).
\end{align}
Clearly, $\lV u_{N,A}\rV\le 1$, and $0\le  u_{N,A}\le 1$.
Then we have
\begin{align}
	\begin{split}
&\lV
u_{N,A}-\lmk
1+h(N)(A^*A)
\rmk^{-1}
h(N)(A^*A)
\rV\\
&=
\lV
\lmk
1+h(N)\bbE_N(A^*A)
\rmk^{-1}
h(N)\bbE_N(A^*A)
-\lmk
1+h(N)(A^*A)
\rmk^{-1}
h(N)(A^*A)
\rV\\
&=\lV
\lmk
1+h(N)\bbE_N(A^*A)
\rmk^{-1}
-\lmk
1+h(N)(A^*A)
\rmk^{-1}
\rV\\
&=
\lV
\lmk
1+h(N)\bbE_N(A^*A)
\rmk^{-1}
\lmk
h(N)\lmk
A^*A-\bbE_N(A^*A)
\rmk
\rmk\lmk
1+h(N)(A^*A)
\rmk^{-1}
\rV\\
&\le
h(N) {f_1(N)}\lV A^*A\rV_{f_1}\to 0,\quad N\to\infty, 
	\end{split}
\end{align}
from (\ref{hasy}).
As $\lmk
1+h(N)(A^*A)
\rmk^{-1}
h(N)(A^*A)\in\caL_s$, we obtain
(\ref{psz}), (\ref{psz2}).
We also have
\begin{align}
	\begin{split}
&\lV
A(1-u_{N,A})
\rV^2
\le
\lV
(1-u_{N,A})(A^*A-\bbE_N(A^*A))(1-u_{N,A})
\rV
+
\lV
(1-u_{N,A})(\bbE_N(A^*A))(1-u_{N,A})
\rV\\
&\le
\lV A^*A\rV_{f_1} {f_1(N)}
+\lV
(1+h(N)\bbE_N(A^*A))^{-1}\bbE_N(A^*A)(1+h(N)\bbE_N(A^*A))^{-1}
\rV\\
&=
\lV A^*A\rV_{f_1} {f_1(N)}
+\frac{1}{h(N)}\lV
(1+h(N)\bbE_N(A^*A))^{-1}h(N)\bbE_N(A^*A)(1+h(N)\bbE_N(A^*A))^{-1}
\rV\\
&\le
\lV A^*A\rV_{f_1} {f_1(N)}
+\frac1{h(N)}=:\varepsilon_N^2.
	\end{split}
\end{align}
For $M>N$, we have
\begin{align}
	\begin{split}
&\frac{\lV
A(1-u_{N,A})
-
\bbE_M\lmk
A(1-u_{N,A})
\rmk
\rV}
{g(M)}
=
 \frac{\lV
\lmk
A-\bbE_M\lmk
A\rmk\rmk
\lmk 1-u_{N,A}\rmk
\rV}
{g(M)}\\
&\le
\lV
A
\rV_{f_{1}}\sup_{M>N}
\lmk\frac{f_1(M)}{g(M)}\rmk=:\varepsilon^{'}_N\to 0,\quad N\to \infty.
	\end{split}
\end{align}
For $M\le N$,
we have 
\begin{align}
&\frac{\lV
A(1-u_{N,A})
-
\bbE_M\lmk
A\lmk 1-u_{N,A}\rmk
\rmk
\rV}
{g(M)}
\le
\frac{2\lV
A(1-u_{N,A})
\rV}{g(N)}\frac{g(N)}{g(M)}
\le
\frac{2\lV
A(1-u_{N,A})
\rV}{g(N)}
\le\frac{2\varepsilon_N}{g(N)}\to0,\quad N\to\infty,
\end{align}
from (\ref{hasy}) and $0<\beta_4<\beta_2<1$.
Hence we obtain,
\begin{align}
\lV A(1-u_{N,A})\rV_g\to \infty,
\end{align}
 proving the Lemma.
\end{proof}
Now we can prove Proposition \ref{iff}.

\begin{proofof}[Proposition \ref{iff}]

Fix $A\in\caD_f$, and $s\in[0,1]$.
By Lemma \ref{lld}, $I_s(A)-\varphi_s(A)\unit\in\caL_s\cap \caL_s^*\cap\caD_{f_1}$.
Applying Lemma \ref{una} to $\lmk I_s(A)-\varphi_s(A)\unit\rmk^*\in\caL_s\cap \caL_s^*\cap\caD_{f_1}$
we obtain a sequence 
$u_N\in \caA_{\Lambda_N}$, $N\in\bbN$ such that
$\lV u_N\rV\le 1$
\begin{align}
&\lV\lmk1-u_N\rmk^*
\lmk I_s(A)-\varphi_s(A)\unit\rmk
\rV_g
=\lV
\lmk I_s(A)-\varphi_s(A)\unit\rmk^*\lmk1-u_N\rmk
\rV_g
\to 0,\label{iv}\\
&{\mathrm{dist}}(u_N, \caL_s)\to 0,\label{dul}
\end{align}
as $N\to 0$.
Applying Lemma \ref{dff} to $u_N\in \caD_f$ and $A\in \caD_f$, we have
\begin{align}\label{dffa}
&\dot{\varphi_s}\lmk u_N^*\lmk I_s(A)-\varphi_s(A)\unit\rmk\rmk\nonumber\\
&=-
\int dt \;\omega_\gamma(t)\int_0^t du 
\varphi_s\lmk
u_N^* \tau_{{\Phi(s)}}^{t-u}\circ
\delta_{\dot{\Phi}(s)}\circ \tau_{\Phi(s)}^{u}(A)
\rmk+
\varphi_s(u_N^*)\dot{\varphi_s}(A).
\end{align}
By (\ref{dul}), 
we have $\lim_{N\to \infty}\varphi_s\lmk
u_N^* \tau_{{\Phi(s)}}^{t-u}\circ
\delta_{\dot{\Phi}(s)}\circ \tau_{\Phi(s)}^{u}(A)
\rmk=0$.
On the other hand, from {\it 2.,} and {\it 4.,} of Lemma \ref{pre}, since $\lV u_N\rV\le 1$, we have, as in (\ref{ttta}), the bound
\begin{align}
\lv
\varphi_s\lmk
u_N^* \tau_{{\Phi(s)}}^{t-u}\circ
\delta_{\dot{\Phi}(s)}\circ \tau_{\Phi(s)}^{u}(A)
\rmk
\rv
\le
C_{f_2,\zeta}^{(1)}b_{f,f_1}(|u|)
\lmk 1+\sup_N\frac{f_1(N)}{f_2(N)}\rmk\lV A\rV_{f}<\infty.
\end{align}
From {\it 2.} of Lemma \ref{pre}, 
\begin{align}\label{recycle}
\int dt \;\omega_\gamma(t)\int_{[0,t]} du 
b_{f,f_1}(|u|)<\infty.
\end{align}
Therefore, by Lebesgue's convergence theorem, we have
\begin{align}
\lim_{N\to\infty}\int dt \;\omega_\gamma(t)\int_0^t du 
\varphi_s\lmk
u_N^* \tau_{{\Phi(s)}}^{t-u}\circ
\delta_{\dot{\Phi}(s)}\circ \tau_{\Phi(s)}^{u}(A)
\rmk=0.
\end{align}
We also have $\lim_{N\to\infty}\varphi_s(u_N^*)\dot{\varphi_s}(A)=0$,
from (\ref{dul}).
Therefore, the right hand side of (\ref{dffa}) goes to $0$ as $N\to\infty$.
The left hand side of  (\ref{dffa}) goes to 
$\dot{\varphi_s}\lmk \lmk I_s(A)-\varphi_s(A)\unit\rmk\rmk$
as $N\to\infty$, because of the continuity (\ref{dcon}) of $\dot\varphi_s$
and (\ref{iv}).
Clearly, $\dot\varphi_s(\unit)=0$.
Therefore, we obtain
$\dot{\varphi_s} \lmk I_s(A)\rmk=0$.

\end{proofof}
\section{Technical Lemmas}\label{lems}
In this Section, we prove various lemmas used in this paper. 
We assume
(i), (ii), (iii) of
Assumption \ref{assump} throughout this section. 
For $t\in\bbR$, 
$[t]$ indicates the largest integer less than or equal to $t$.

\subsection{Properties of $\tau_{\Phi(s)}$}
First we recall several facts from \cite{BMNS} and \cite{NSY}.
Define positive functions $F(r)$ and 
 $F_1(r)$ on $\bbR_{\ge 0}$ by $F(r):=(1+r)^{-(\nu+1)}$, $F_1(r):=(1+r)^{-(\nu+1)}e^{-r}$.
For a path of interactions satisfying Assumption \ref{assump},
there exist
positive constants
$C_{1}'$, $v$ satisfying the following Lieb-Robinson bound:
For any $X,Y\in{\mathfrak S}_{\bbZ^\nu}$, $A\in\caA_X$,
$B\in\caA_Y$, $\Lambda\in{\mathfrak S}_{\bbZ^\nu}$, $s\in[0,1]$ and $t\in\bbR$, we have
\begin{align}\label{lr}
\lV
\left[
\tau^{t}_{\Phi(s)}(A),B
\right]
\rV,\quad
\lV
\left[
\tau^{t}_{\Phi(s),\Lambda}(A),B
\right]
\rV
\le
C_{1}'e^{v|t|}
\sum_{x\in X,y\in Y}F_1(d(x,y))\lV A\rV \lV B\rV.
\end{align}
We fix the constant $v$ and call it the Lieb-Robinson velocity.
From this and Corollary 4.4. of \cite{NSY} (Proposition \ref{NSY44}) we obtain the following.
\begin{lem}\label{41lem}
There is a positive constant $C_1>0$ such that
\begin{align}\label{pen}
\lV
\tau^t_{\Phi(s),\Lambda}(A)-\bbE_N\lmk
\tau^{t}_{\Phi(s),\Lambda}(A)
\rmk
\rV,\;
\lV
\tau^{t}_{\Phi(s)}(A)-\bbE_N\lmk
\tau^{t}_{\Phi(s)}(A)
\rmk
\rV
\le
C_{1}\lv \Lambda_M\rv e^{v|t|-(N-M)}\lV A\rV,
\end{align}
for any $M,N\in\nan$ with $M\le N$, $A\in\caA_{\Lambda_M}$ and $\Lambda\in \mathfrak{S}_{\bbZ^\nu}$.
\end{lem}
We also have the following (see Corollary 3.6 (3.80) of \cite{NSY}.)
\begin{lem}\label{locl}
There is a constant $C_4>0$  such that
\begin{align}
\sup_{s\in[0,1]}\lV
\tau_{\Phi(s),\Lambda_n}^{-u}(B)-\tau_{\Phi(s)}^{-u}(B)
\rV
\le C_4 \lv \Lambda_M\rv|u|
e^{|u|v-(n-M)}\lV B\rV,\quad n\ge M,\quad u\in\bbR,\quad B\in\caA_{\Lambda_M}.
\end{align}

\end{lem}
It is standard to derive the  following from Lemma \ref{locl} (cf. \cite{BR1}).
\begin{lem}\label{ao}
For any $A\in\caA$,
\begin{align}
\sup_{s\in[0,1]}\lV
\tau_{\Phi(s),\Lambda_n}^{-u}(A)
-\tau_{\Phi(s)}^{-u}(A)
\rV
\to 0,
\end{align}
uniformly in compact $u\in \bbR$.
In particular, for each $A\in\caA$,
$\bbR\times [0,1]\ni (u,s)\to \tau_{\Phi(s)}^{-u}(A)\in \caA$ is 
continuous with respect to the norm $\lV \cdot\rV$.
\end{lem}

\begin{lem}\label{44}
Suppose $f_1,f_2: (0, \infty) \to (0,\infty)$ are continuous decreasing functions with $\lim_{t\to \infty} f_i(t) =0$, for $i=1,2$. Suppose that we have
	\begin{equation}\label{eq:f-decay}
		\begin{split}
\lim_{N\to\infty}\lmk
 \frac{|\Lambda_{\lcm \frac N2\rcm}|  e^{-(N-\lcm \frac N2\rcm)}}{f_2(N)}\rmk=0.
		\end{split}
	\end{equation}
and
\begin{align}
\lim_{N\to\infty}\frac{f_1\lmk \lcm \frac N2\rcm\rmk}{f_2(N)}=0.
\end{align}
Then
\begin{align}
\sup_{s\in[0,1]}\lV
\tau_{\Phi(s),\Lambda_n}^{-u}(A)
-\tau_{\Phi(s)}^{-u}(A)
\rV_{f_2}
\to 0,\quad A\in \caD_{f_1},
\end{align}
uniformly in compact $u\in \bbR$.
In particular, for each $A\in\caD_{f_1}$,
$\bbR\times [0,1]\ni (u,s)\to \tau_{\Phi(s)}^{-u}(A)\in \caD_{f_2}$ is 
continuous with respect to the norm $\lV \cdot\rV_{f_2}$.
\end{lem}
\begin{proof}
Let $A\in\caD_{f_1}$.
From Lemma \ref{ao}, we have
\begin{align}
\sup_{s\in[0,1]}\lV
\tau_{\Phi(s),\Lambda_n}^{-u}(A)
-\tau_{\Phi(s)}^{-u}(A)
\rV
\to 0.
\end{align}
Applying Lemma \ref{locl},
 for $N\le [\frac{n}{2}]$, we have
\begin{align}
	\begin{split}
&\lV
\tau_{\Phi(s),\Lambda_n}^{-u}(A)-\tau_{\Phi(s)}^{-u}(A)
-\bbE_N\lmk
\tau_{\Phi(s),\Lambda_n}^{-u}(A)-\tau_{\Phi(s)}^{-u}(A)
\rmk
\rV\\
&\le
\lV
\tau_{\Phi(s),\Lambda_n}^{-u}\lmk \bbE_{[\frac n2]}\lmk A\rmk\rmk
-\tau_{\Phi(s)}^{-u}\lmk \bbE_{[\frac n2]}\lmk A\rmk\rmk
-\bbE_N\lmk
\tau_{\Phi(s),\Lambda_n}^{-u}\lmk
\bbE_{[\frac n2]}\lmk A\rmk\rmk-\tau_{\Phi(s)}^{-u}\lmk
\bbE_{[\frac n2]}\lmk A\rmk\rmk
\rmk
\rV\\
&+4\lV \bbE_{[\frac n2]}\lmk A\rmk-A\rV\\
&\le
2C_4 \lv \Lambda_{[\frac n2]}\rv|u|
e^{|u|v-(n-[\frac n2])}\lV A\rV+4 f_1\lmk \lcm\frac n2\rcm\rmk\lV A\rV_{f_1}
	\end{split}
\end{align}
On the other hand,
from Lemma \ref{41lem}
$N\ge \lcm\frac n2\rcm$,
\begin{align}
	\begin{split}
&\lV
\tau_{\Phi(s),\Lambda_n}^{-u}\lmk \bbE_{[\frac N2]}\lmk A\rmk\rmk
-\bbE_N\lmk
\tau_{\Phi(s),\Lambda_n}^{-u}\lmk
\bbE_{[\frac N2]}\lmk A\rmk\rmk\rmk
\rV
\le  C_1 \norm{A} |\Lambda_{\lcm \frac N2\rcm}|  e^{v|u|-(N-\lcm \frac N2\rcm)},\\
&\lV
\tau_{\Phi(s)}^{-u}\lmk \bbE_{[\frac N2]}\lmk A\rmk\rmk
-\bbE_N\lmk
\tau_{\Phi(s)}^{-u}\lmk
\bbE_{[\frac N2]}\lmk A\rmk\rmk\rmk
\rV
\le  C_1 \norm{A} |\Lambda_{\lcm \frac N2\rcm}|  e^{v|u|-(N-\lcm \frac N2\rcm)}.
	\end{split}
\end{align}
Therefore, for $N\ge \lcm\frac n2\rcm$, we have
\begin{align}
	\begin{split}
&\lV
\tau_{\Phi(s),\Lambda_n}^{-u}(A)-\tau_{\Phi(s)}^{-u}(A)
-\bbE_N\lmk
\tau_{\Phi(s),\Lambda_n}^{-u}(A)-\tau_{\Phi(s)}^{-u}(A)
\rmk
\rV\\
&\le\lV
\tau_{\Phi(s),\Lambda_n}^{-u}\lmk \bbE_{[\frac N2]}\lmk A\rmk\rmk
-\tau_{\Phi(s)}^{-u}\lmk \bbE_{[\frac N2]}\lmk A\rmk\rmk
-\bbE_N\lmk
\tau_{\Phi(s),\Lambda_n}^{-u}\lmk
\bbE_{[\frac N2]}\lmk A\rmk\rmk-\tau_{\Phi(s)}^{-u}\lmk
\bbE_{[\frac N2]}\lmk A\rmk\rmk
\rmk
\rV\\
&+4\lV \bbE_{[\frac N2]}\lmk A\rmk-A\rV\\
&\le
2C_1 \norm{A} |\Lambda_{\lcm \frac N2\rcm}|  e^{v|u|-(N-\lcm \frac N2\rcm)}
+4f_1\lmk\lcm\frac N2\rcm\rmk\lV A\rV_{f_1}
	\end{split}
\end{align}
Hence we obtain
\begin{align}
\sup_{s\in[0,1]}\lV
\tau_{\Phi(s),\Lambda_n}^{-u}(A)
-\tau_{\Phi(s)}^{-u}(A)
\rV_{f_2}
\le
&\max\left\{
\begin{gathered}
2C_4|u|
e^{|u|v}\lV A\rV \frac{\lv \Lambda_{[\frac n2]}\rv e^{-(n-[\frac n2])}}{f_2\lmk\lcm \frac n2\rcm\rmk}
+4 \frac{f_1\lmk [\frac n2]\rmk}{f_2\lmk\lcm \frac n2\rcm\rmk}\lV A\rV_{f_1},
\\
2C_1 \norm{A}\sup_{N\ge \lcm\frac n2\rcm}\lmk
 \frac{|\Lambda_{\lcm \frac N2\rcm}|  e^{v|u|-(N-\lcm \frac N2\rcm)}}{f_2(N)}\rmk
+\sup_{N\ge \lcm\frac n2\rcm}\lmk\frac{4f_1\lmk\lcm\frac N2\rcm \rmk\lV A\rV_{f_1}}{f_2(N)}\rmk
\end{gathered}
\right\}\nonumber\\
&+\sup_{s\in[0,1]}\lV
\tau_{\Phi(s),\Lambda_n}^{-u}(A)
-\tau_{\Phi(s)}^{-u}(A)
\rV,
\end{align}
and 
$\sup_{s\in[0,1]}\lV
\tau_{\Phi(s),\Lambda_n}^{-u}(A)
-\tau_{\Phi(s)}^{-u}(A)
\rV_{f_2}$ converges to $0$ as $n\to\infty$, uniformly in compact $u$.

\end{proof}
\begin{lem}\label{48}
Let $f,f_1:(0,\infty)\to (0,\infty)$
be continuous decreasing functions 
with $\lim_{t\to\infty}f(t)=0$.
Suppose that
\begin{align}\label{eq:weight-req}
	\begin{split}
&\int_{4v|t|\ge 1} dt \omega_\gamma(t) \frac{2|t|}{f_1(4v|t|)}<\infty,  \\
&\sup_{N\in\nan}\lmk\frac{f(N-{\lcm \frac N2\rcm})}{f_1(N)}\rmk<\infty,\\
 &\sup_{N\in\nan}\lmk
\frac{|\Lambda_N| e^{- \frac{
{\lcm \frac N2\rcm}}2}}{f_1(N)}\rmk<\infty.
	\end{split}
\end{align}

Then $\tau_{\Phi(s)}^t\lmk\caD_f\rmk\subset\caD_{f_1}$ and
there is a non-negative non-decreasing function on $\bbR_{\ge 0}$, $b_{f,f_1}(t)$ such that
\begin{align}\label{b1}
\int\; dt\; \omega_\gamma(t)\; |t|\cdot b_{f,f_1}(|t|)<\infty.
\end{align}
\begin{align}\label{b1t}
&
\sup_{n\in\nan}\sup_{s\in[0,1]}\lV
\tau_{\Phi_{n}(s)}^t\lmk
A
\rmk
\rV_{f_1}
,
\sup_{s\in[0,1]}\lV
\tau_{\Phi(s)}^t\lmk
A
\rmk
\rV_{f_1}
\le
b_{f,f_1}(|t|)\lV A\rV_f,\quad A\in\caD_f.
\end{align}
\end{lem}
\begin{proof}
Let $A\in\caD_f$.
We have to estimate 
\begin{align}\label{tn}
\frac{\lV \tau^t _{\Phi(s)}(A)-\bbE_N\lmk \tau^t _{\Phi(s)}(A)\rmk\rV}{f_1(N)},\quad N\in\nan.
\end{align}
From Lemma \ref{41lem} for $A\in\caD_f$, $N,k\in\nan$ with $k<N$, we obtain
\begin{align}
	\begin{split}
&\lV
\tau^t _{\Phi(s)}(A)-\bbE_N\lmk \tau^t _{\Phi(s)}(A)\rmk
\rV
\le 
\lV
\tau^t _{\Phi(s)}\lmk\bbE_k(A)\rmk-\bbE_N\lmk \tau^t _{\Phi(s)}\lmk\bbE_k(A)\rmk\rmk
\rV
+2\lV A-\lmk\bbE_k(A)\rmk\rV\\
&\le
2 \lV A\rV_f f(k)+C_1 \norm{A} |\Lambda_k| e^{v|t|- (N-k)}.
	\end{split}
\end{align}
For $N\in\nan$ with $4 v\lv t\rv \le N$, we use this bound with
$k:=N-{\lcm \frac N2\rcm}$
to estimate (\ref{tn}).
Then we have
\begin{align}
	\begin{split}
\lV
\tau^t _{\Phi(s)}(A)-\bbE_N\lmk \tau^t _{\Phi(s)}(A)\rmk
\rV
\le
2 \lV A\rV_f \lmk f(N-{\lcm \frac N2\rcm})\rmk+ C_1\norm{A}|\Lambda_N| e^{v|t|- \lcm
{\frac N2}\rcm}\\
\le
2 \lV A\rV_f \lmk f(N-{\lcm \frac N2\rcm})\rmk+ C_1 \norm{A} |\Lambda_N| e^{- \frac{
{\lcm \frac N2\rcm}}2+\frac 12}.
	\end{split}
\end{align}
On the other hand, for $N\in\nan$ with $ 4 v\lv t\rv> N$, we simply have
\begin{align}
\lV
\tau^t_{\Phi(s)}(A)-\bbE_N\lmk \tau^t _{\Phi(s)}(A)\rmk
\rV\le 2\lV A\rV.
\end{align}
Hence we obtain
\begin{align}\label{eq:b1function}
\lV \tau^t _{\Phi(s)}(A)\rV_{f_1}
\le 
\lmk 1+
\max\left\{
\begin{gathered}
2
\sup_{N\in\nan}\lmk\frac{\lmk f(N-{\lcm \frac N2\rcm})\rmk}{f_1(N)}\rmk
+ C_1 \sup_{N\in\nan}\lmk
\frac{|\Lambda_N| e^{- \frac{
{\lcm \frac N2\rcm}}2+\frac 12}}{f_1(N)}\rmk,\\
\frac{2}{f_1\lmk \lmk 4 v\lv t\rv\rmk\rmk}\unit_{ 4v\lv t\rv\ge 1}
\end{gathered}
\right\}\rmk \lV A\rV_f
=:b_{f,f_1}(t) \lV A\rV_f,
\end{align}
for $A\in\caD_f$ and $t\in\bbR$, $s\in[0,1]$.
Here $\unit_{ 4v\lv t\rv\ge 1}$ is the characteristic function for
$\{t\in \bbR\mid 4v\lv t\rv\ge 1\}$.
From the assumptions and (\ref{omeganom}), $b_{f,f_1}(t)$ satisfies the required condition.
The inequality for $\tau^t _{\Phi_{n}(s)}(A)$ can be proven in the same way.
\end{proof}
\begin{lem} \label{410}
Let $f,f_1:(0,\infty)\to (0,\infty)$
be continuous decreasing functions 
with $\lim_{t\to\infty}f(t)=\lim_{t\to\infty} f_1(t)=0$.
Suppose that
\begin{align}
	\begin{split}
&\sup_{N\in\nan}\frac{f(N-\lcm { \frac N2}\rcm)}{f_1(N)}<\infty,\\
&\sup_{N\in\nan}\frac{\lv\Lambda_N\rv e^{-\frac{\lcm { \frac N2}\rcm}2}}{f_1(N)}<\infty,\\
&\sup_{N\in\nan}\frac{W_\gamma\lmk \frac{\lcm { \frac N2}\rcm}{2v}\rmk}{f_1(N)}<\infty.
	\end{split}
\end{align}
(Recall (\ref{Wint}).) For $s\in[0,1]$ and $A\in\caA$, we set
\begin{align}
I_s(A):=
\int dt \;\omega_\gamma(t) \tau_{\Phi(s)}^t(A) .
\end{align}
The integral can be understood as a Bochner integral of $(\caA, \lV\cdot\rV)$.
Then for any $A \in \caD_f$ and $s\in[0,1]$, we have $I_s(A) \in \caD_{f_1}$. 
\end{lem}
\begin{proof}
That the integral can be understood as a Bochner integral of $(\caA, \lV\cdot\rV)$
is from the continuity of 
$\bbR\ni t\to \tau_{\Phi(s)}^t(A) \in\caA$, Lemma \ref{ao} and
$\omega_\gamma\in L^1(\bbR)$.

From (\ref{pen}), we obtain
\begin{equation}\label{tauloc}
		\begin{split}
& \norm{ \tau^t _{\Phi(s)}(\EE_{k}(A))-\bbE_N
\lmk  \tau^t _{\Phi(s)}(\EE_{k}(A))\rmk } 
 \leq C_{1}\lv \Lambda_k\rv e^{v|t|-(N-k)}\lV A\rV,
		\end{split}
	\end{equation}
for any $A\in\caD_f$, $s\in[0,1]$, $t\in\bbR$, $N,k\in\nan$, with $k\le N$.

For any $A\in\caD_f$, $s\in[0,1]$, $N\in\nan$,  we have
\begin{align}
	\begin{split}
&\lV
I_s(A)-\bbE_N\lmk I_s(A)\rmk
\rV\\
&\le
\lV
I_s\lmk \bbE_{N-\lcm { \frac N2}\rcm}(A)\rmk-\bbE_N\lmk I_s\lmk \bbE_{N-\lcm { \frac N2}\rcm}(A)\rmk\rmk
\rV
+2\lV
A-\bbE_{N-[{ \frac N2}]}(A)
\rV\\
&\le
\int_{|t|\le \frac{\lcm { \frac N2}\rcm}{2v}}dt\omega_\gamma(t) 
 \norm{ \tau^t _{\Phi(s)}(\EE_{N-\lcm { \frac N2}\rcm}(A))-\bbE_N
\lmk  \tau^t _{\Phi(s)}(\EE_{N-\lcm { \frac N2}\rcm}(A))\rmk } \\
&+\int_{|t|\ge \frac{\lcm { \frac N2}\rcm}{2v}}dt\omega_\gamma(t) 
\norm{ \tau^t _{\Phi(s)}(\EE_{N-\lcm { \frac N2}\rcm}(A))-\bbE_N
\lmk  \tau^t _{\Phi(s)}(\EE_{N-\lcm { \frac N2}\rcm}(A))\rmk } 
+2\lV A\rV_f f(N-\lcm { \frac N2}\rcm)\\
&\le
\int_{|t|\le \frac{\lcm { \frac N2}\rcm}{2v}}dt\omega_\gamma(t) 
C_{1}\lv \Lambda_N\rv e^{v|t|-[{ \frac N2}]}\lV A\rV
+\int_{|t|\ge \frac{\lcm { \frac N2}\rcm}{2v}}dt\omega_\gamma(t) 2\lV A\rV
+2\lV A\rV_f f(N-\lcm { \frac N2}\rcm)\\
&\le
C_{1}\lv \Lambda_N\rv e^{-\frac{[{ \frac N2}]}2}\lV A\rV
+ 4\lV A\rV W_\gamma \lmk \frac{\lcm { \frac N2}\rcm}{2v}\rmk
+2\lV A\rV_f f(N-\lcm { \frac N2}\rcm).
	\end{split}
\end{align}
For the first and the fourth inequality, we used (\ref{omeganom}).
We used (\ref{tauloc}), with $k=N-[\frac N2]$, for the third inequality.

Hence we obtain
\begin{align}
	\begin{split}
&\sup_{N\in\nan}
\frac{\lV
I_s(A)-\bbE_N\lmk I_s(A)\rmk
\rV}{f_1(N)}\\
&\le
 C_1 \norm{A} \sup_{N\in\nan}\frac{ |\Lambda_N| e^{- \frac{\lcm { \frac N2}\rcm}{2}}}{f_1(N)}
+ 4\lV A\rV \sup_{N\in\nan}\frac{W_\gamma\lmk \frac{\lcm { \frac N2}\rcm}{2v}\rmk}{f_1(N)}
+2\lV A\rV_f \sup_{N\in\nan}\frac{f(N-\lcm { \frac N2}\rcm)}{f_1(N)}<\infty,
	\end{split}
\end{align}
for any $A\in\caD_f$ and $s\in[0,1]$.
Hence we obtain $I_s(\caD_f)\subset \caD_{f_1}$, for any $s\in[0,1]$.
\end{proof}

\subsection{Estimates on $\alpha_s$}
In the following, we prove estimates on quasi-locality of the automorphisms $\alpha_s$ and $\alpha_{s,\Lambda}$. To do this, we first recall a theorem from \cite{BMNS} on Lieb-Robinson bounds. 

Define $\tilde{h}(x) =   \frac{x}{\ln^2(x)}$ for $x>1$. Define the weight function as:
	\begin{equation*}
		\begin{split}
h(x) = \bigg{ \{ }\begin{array}{l l} \tilde{h}(e^2) & \text{ if } 0 \leq x \leq e^2 \\ \tilde{h}(x) & \text{ otherwise } \end{array}
		\end{split}.
	\end{equation*}

The Lieb-Robinson bound for the automorphisms $\alpha_s$ is given as follows:
there exists a constant $C_2>0$, $\eta_1>0$, $\tilde a>0$ satisfying the following:
setting $\hat h(x):=\eta_1h(\tilde a x)$,  we have
	\begin{equation}\label{lra}
		\begin{split}
\norm{ [\alpha_s(B),A] }, \norm{ [\alpha_{s,\Lambda_{n}}(B),A] } \leq 
\frac{C_2}2 \norm{A} \norm{B}|X|   e^{-\hat h( d(X,Y))}
		\end{split}
	\end{equation} 
for any $A \in \A_X$, $B \in \A_Y$ with $X,Y\in {\mathfrak S}_{\bbZ^\nu}$, and $s\in[0,1]$. 
See Theorem 4.5 of \cite{BMNS} and Corollary 6.14 of \cite{NSY}.
(Note that in \cite{BMNS}, Assumption 4.3 about a spectral gap is assumed but for the proof of (\ref{lra}), this assumption is not used.)
From Corollary 3.6 (3.80) of \cite{NSY}, there is a constant $C_3>0$ such that
\begin{align}\label{faa}
\sup_{s\in[0,1]}\lV
\alpha_{s,\Lambda_n}^{-1}(A)-\alpha_{s}^{-1}(A)
\rV
\le C_3 \lv \Lambda_M\rv
e^{-\hat h\lmk n-M \rmk}\lV A\rV,\quad n\ge M,\quad M\in\nan,\quad \text{and }
A\in\caA_{\Lambda_M}.
\end{align}
From (\ref{lra}), we obtain the following.
\begin{lem}\label{n43}
For any $M,N\in\nan$ with $M<N$, we have
\begin{equation}\label{ail}
{\norm{ \alpha_s^{-1}(A) - \m{E}_{ N} (\alpha_s^{-1}(A)) }} \leq
C_2\lv
\Lambda_M\rv \lV A\rV
{e^{-\hat h\lmk {N-M}\rmk}},\quad A\in\caA_{\Lambda_M}.
\end{equation}
\end{lem}
\begin{proof}
If $A\in\caA_{\Lambda_M}$ and $B \in \A _{ \Lambda_{N}^c}$, 
then $B = \lim_{n\to \infty} B_n $ in norm for a sequence of local observables 
$B_n\in \A _{\Lambda_{N}^c}\cap \Aloc$ and:
	\begin{equation}
		\begin{split}
\norm{[B,\alpha_s^{-1}(A)]}=\norm{ [\alpha_s(B),A] } & \leq \limsup_n\lmk
 2\norm{A} \norm{B-B_n} + \frac{C_2}2 \norm{A} |\Lambda_M| 
  \norm{B_n} e^{-\hat h(N-M)}\rmk\\
& =  \frac{C_2}2 |\Lambda_M| \norm{A}  \norm{B} e^{-\hat h(N-M)}.
		\end{split}
	\end{equation}
And so by Corollary 4.4. of \cite{NSY} (Proposition \ref{NSY44}) we conclude (\ref{ail}).
\end{proof}
From this Lemma we immediately obtain the following:
\begin{lem}\label{42}
Suppose $f: (0, \infty) \to (0,\infty)$ is a continuous decreasing function with $\lim_{t\to \infty} f(t) =0$. Suppose that for all $M\in\nan$, we have
	\begin{equation}
		\begin{split}
\sup _n \frac{ e^{-\hat h(n)} }{f(M+n)} < \infty
		\end{split}
	\end{equation}
 then $\alpha_s^{-1}( \Aloc ) \subset D_f$. 
\end{lem}
\begin{proof}
%
%
Let $M\in\nan$ and $A\in\caA_{\Lambda_M}$.
From (\ref{ail}), we have
	\begin{equation}\label{ali}
		\begin{split}\sup_{R\in\nan}\lmk
\frac{\norm{ \alpha_s^{-1}(A) - \m{E}_{ {M+R}} (\alpha_s^{-1}(A)) }}{f(M+R)}\rmk
 \leq \sup_{R\in\nan}
\lmk
C_2\lv \Lambda_M\rv\frac{e^{-\hat h(R)}}{f(M+R)}\rmk\lV A\rV<\infty.
		\end{split}
	\end{equation}  
	Hence we obtain $ \alpha_s^{-1}(A)\in\caD_f$.
\end{proof}

\begin{lem}\label{nn46}
Let $f_1,f_2:(0,\infty)\to (0,\infty)$
be continuous decreasing functions 
with $\lim_{t\to\infty}f_i(t)=0$, $i=1,2$.
Suppose that 

\begin{align}
	\begin{split}
\sup_{N\in\nan}\lmk \frac{f_1\lmk N-{\lcm \frac N2\rcm}\rmk}{f_2(N)}\rmk<\infty,\\
\sup_{N\in\nan}\lmk\frac{e^{-\hat h\lmk {\lcm \frac N2\rcm}\rmk}\lv\Lambda_{N-{\lcm \frac N2\rcm}}\rv}{f_2(N)}\rmk<\infty.
	\end{split}
\end{align}
Then we have $\alpha_s^{-1}(\caD_{f_1})\subset \caD_{f_2}$,
$\alpha_{s,\Lambda}^{-1}(\caD_{f_1})\subset \caD_{f_2}$
for any $s\in[0,1]$ , and $\Lambda\in{\mathfrak S}_{\bbZ^{\nu}}$.
Furthermore we have the following inequalities:
\begin{align}
\sup_{s\in[0,1]}\lV
\alpha_s^{-1}(A)
\rV_{f_2},
\sup_{s\in[0,1]}\lV
\alpha_{s,\Lambda}^{-1}(A)
\rV_{f_2}
\le
\lV
A
\rV_{f_1}\lmk
1+
\sup_{N\in\nan}\lmk
\frac{2f_1\lmk N-{\lcm \frac N2\rcm}\rmk+
C_2 
{e^{-\hat h\lmk {\lcm \frac N2\rcm}\rmk}\lv\Lambda_{N-{\lcm \frac N2\rcm}}\rv}}
{f_2(N)}
\rmk\rmk,
\end{align}
for any $A\in\caD_{f_1}$.

\end{lem}
\begin{proof}
This follows from the following inequality: for each $N\in\nan$ and 
$A\in\caD_{f_1}$,
\begin{align}
	\begin{split}
&\lV
\alpha_s^{-1}(A)-\bbE_N\lmk \alpha_s^{-1}(A)\rmk
\rV 
\\
&\le
\lV
\alpha_s^{-1}\lmk A-\bbE_{N-{\lcm \frac N2\rcm}}(A)\rmk-\bbE_N\lmk \alpha_s^{-1}
\lmk A-\bbE_{N-{\lcm \frac N2\rcm}}(A)\rmk\rmk
\rV\\
&+\lV
\alpha_s^{-1}\lmk \bbE_{N-{\lcm \frac N2\rcm}}(A)\rmk-\bbE_N\lmk \alpha_s^{-1}\lmk\bbE_{N-{\lcm \frac N2\rcm}}(A)\rmk\rmk
\rV\\
&\le
\lV A\rV_{f_1}
\lmk
2f_1\lmk N-{\lcm \frac N2\rcm}\rmk+
C_2 
{e^{-\hat h\lmk {\lcm \frac N2\rcm}\rmk}\lv\Lambda_{N-{\lcm \frac N2\rcm}}\rv}
\rmk.
	\end{split}
\end{align}
\end{proof}

\begin{lem}\label{43}
Suppose $f: (0, \infty) \to (0,\infty)$ is a continuous decreasing function with $\lim_{t\to \infty} f(t) =0$. Suppose that for all $M\in\nan$, we have
	\begin{equation}
		\begin{split}
\lim_{n\to\infty}\sup_{N\ge n}\lmk \frac{e^{-\hat h\lmk N-M\rmk}}{f(N)}\rmk=0.
		\end{split}
	\end{equation}
	Then we have
\begin{align}
\sup_{s\in[0,1]}\lV
\alpha_{s,\Lambda_n}^{-1}(A)-\alpha_{s}^{-1}(A)
\rV_f\to 0,\quad A\in\caA_{\rm loc}.
\end{align}
In particular, for each $A\in\caA_{\rm loc}$, $\bbR\ni s\to \alpha_{s}^{-1}(A)\in \caD_{f}$ is 
continuous with respect to the norm $\lV \cdot\rV_{f}$.
\end{lem}
\begin{proof}
Let $A\in \caA_{\Lambda_M}$. From (\ref{faa}),
 for $n\ge N\ge M$, 
we have
\begin{align}
\sup_{s\in[0,1]}
\frac{\lV
\alpha_{s,\Lambda_n}^{-1}(A)-\alpha_{s}^{-1}(A)-
\bbE_N\lmk
\alpha_{s,\Lambda_n}^{-1}(A)-\alpha_{s}^{-1}(A)
\rmk
\rV}{f(N)}
\le 2C_3 \lv \Lambda_M\rv
\frac{e^{-\hat h\lmk n-M\rmk}}{f(n)}\lV A\rV.
\end{align}
On the other hand, for $M\le n\le N$, from (\ref{ail})
\begin{align}
	\begin{split}
&\sup_{s\in[0,1]}
\frac{\lV
\alpha_{s,\Lambda_n}^{-1}(A)-\alpha_{s}^{-1}(A)-
\bbE_N\lmk
\alpha_{s,\Lambda_n}^{-1}(A)-\alpha_{s}^{-1}(A)
\rmk
\rV}{f(N)}
=\sup_{s\in[0,1]}
\frac{\lV
\alpha_{s}^{-1}(A)-
\bbE_N\lmk
\alpha_{s}^{-1}(A)
\rmk
\rV}{f(N)}\\
&\le 
C_2\lv
\Lambda_M\rv \lV A\rV
\frac{e^{-\hat h\lmk N-M\rmk}}{f(N)}\le
C_2\lv
\Lambda_M\rv \lV A\rV
\sup_{N\ge n}\lmk \frac{e^{-\hat h\lmk N-M\rmk}}{f(N)}\rmk.
	\end{split}
\end{align}
Furthermore, for $n\ge M>N$, we have
\begin{align}
	\begin{split}
&\sup_{s\in[0,1]}
\frac{\lV
\alpha_{s,\Lambda_n}^{-1}(A)-\alpha_{s}^{-1}(A)-
\bbE_N\lmk
\alpha_{s,\Lambda_n}^{-1}(A)-\alpha_{s}^{-1}(A)
\rmk
\rV}{f(N)}\\
&\le
 2C_3 \lv \Lambda_M\rv
\frac{e^{-\hat h(n-M)}}{f(M)}\lV A\rV.
	\end{split}
\end{align}
Hence we obtain
\begin{align}
	\begin{split}
&\sup_{s\in[0,1]}\lV
\alpha_{s,\Lambda_n}^{-1}(A)-\alpha_{s}^{-1}(A)
\rV_f\\
&\le \lV A\rV \lmk
1+\max\left\{
2C_3 \lv \Lambda_M\rv
\frac{e^{-\hat h\lmk n-M \rmk}}{f(n)},\;
C_2\lv
\Lambda_M\rv 
\sup_{N\ge n}\lmk \frac{e^{-\hat h\lmk N-M\rmk}}{f(N)}\rmk,\;
 2C_3 \lv \Lambda_M\rv
\frac{e^{-\hat h(n-M)}}{f(M)}
\right\}
\rmk\to 0,\quad n\to\infty.
	\end{split}
\end{align}
\end{proof}
\begin{lem}\label{418}
Let $f,f_0,f_1: (0, \infty) \to (0,\infty)$ be continuous decreasing functions with $\lim_{t\to \infty} f(t) =\lim_{t\to \infty} f_0(t)=\lim_{t\to \infty} f_1(t)
=0$. Suppose that for all $M\in\nan$, we have
	\begin{equation}
		\begin{split}
\lim_{n\to\infty}\sup_{N\ge n}\lmk \frac{e^{-\hat h\lmk N-M\rmk}}{f(N)}\rmk=0.
		\end{split}
	\end{equation}
Suppose that
\begin{align}
	\begin{split}
&\sup_{N\in\nan}\frac{f_1(N-{\lcm \frac N2\rcm})}{f(N)}<\infty,\\
&\sup_{N\in\nan}\frac{e^{-{\hat h({\lcm \frac N2\rcm})}}}{f(N)}\lv \Lambda_{N-{\lcm \frac N2\rcm}}\rv<\infty.
	\end{split}
\end{align}
Suppose that
\begin{align}
\lim_{N\to\infty}\frac{f_0(N)}{f_1(N)}=0.
\end{align}
Then we have $\alpha_s^{-1}\lmk \caD_{f_0}\rmk\subset \caD_f$ and
\begin{align}
\sup_{s\in[0,1]}\lV
\alpha_{s,\Lambda_n}^{-1}(A)-\alpha_{s}^{-1}(A)
\rV_f\to 0,\quad A\in\caD_{f_0}.
\end{align}
In particular, for each $A\in\caD_{f_0}$, $[0,1]\ni s\to \alpha_{s}^{-1}(A)\in \caD_{f}$ is 
continuous with respect to the norm $\lV \cdot\rV_{f}$.
\end{lem}
\begin{proof}
As 
\begin{align}
\sup_{N\in\nan}\frac{f_0(N)}{f_1(N)}<\infty,
\end{align}
we have $\caD_{f_0}\subset \caD_{f_1}$.
By Lemma \ref{nn46} with $(f_1,f_2)$ replaced by $(f_1,f)$,
we get $\alpha_s^{-1}\lmk \caD_{f_1}\rmk\subset \caD_f$.
Hence we have $\alpha_s^{-1}\lmk \caD_{f_0}\rmk\subset \caD_f$.
For any $A\in\caD_{f_0}$,
\begin{align}
	\begin{split}
&\limsup_{n\to\infty}\sup_{s\in[0,1]}\lV
\alpha_{s,\Lambda_n}^{-1}(A)-\alpha_{s}^{-1}(A)
\rV_f\\
&=
\limsup_{n\to\infty}\sup_{s\in[0,1]}\lV
\alpha_{s,\Lambda_n}^{-1}\lmk
A-\bbE_M(A)
\rmk
-\alpha_{s}^{-1}\lmk
A-\bbE_M(A)
\rmk
+\alpha_{s,\Lambda_n}^{-1}\lmk
\bbE_M(A)
\rmk
-\alpha_{s}^{-1}\lmk
\bbE_M(A)
\rmk
\rV_f\\
&\le
\limsup_{n\to\infty}\sup_{s\in[0,1]}\lV
\alpha_{s,\Lambda_n}^{-1}\lmk\bbE_M(A)\rmk -\alpha_{s}^{-1}\lmk\bbE_M(A)\rmk
\rV_f\\
&+2\lV A-\bbE_M(A)\rV_{f_1}\lmk
\sup_{N\in\nan}\frac{\lmk
2f_1\lmk N-{\lcm \frac N2\rcm}\rmk+
C_2 
{e^{-\hat h({\lcm \frac N2\rcm})}}\lv\Lambda_{N-{\lcm \frac N2\rcm}}\rv
\rmk}{f(N)}
+1\rmk
\\
&=2\lV A-\bbE_M(A)\rV_{f_1}
\lmk
\sup_{N\in\nan}\frac{\lmk
2f_1\lmk N-{\lcm \frac N2\rcm}\rmk+
C_2 
{e^{-\hat h({\lcm \frac N2\rcm})}}\lv\Lambda_{N-{\lcm \frac N2\rcm}}\rv
\rmk}{f(N)}
+1\rmk\to 0,\quad M\to\infty.
	\end{split}
\end{align}
For the inequality, we used Lemma \ref{nn46}.
For the last line we used 
Lemma \ref{43}.
As we have $\lim_{M\to\infty}\lV A-\bbE_M(A)\rV_{f_1}=0$ by Lemma \ref{n47}
with $(f,f_1)$ replaced by $(f_0,f_1)$,
we have proven the claim.
\end{proof}

\subsection{Properties of $\delta_{\Phi(s)}$, $\delta_{\dot\Phi(s)}$}

\begin{lem}\label{f2f3}
Let $f_2: (0, \infty) \to (0,\infty)$ be a continuous decreasing function
such that
 \begin{align}\label{114p}
{\sum_{k=2}^\infty k^\nu f_2(k-1)}<\infty.
\end{align}
Let $f_3: (0, \infty) \to (0,\infty)$ be continuous decreasing function 
with $\lim_{t\to \infty} f_3(t) =0$
such that
\begin{align}\label{f2f3rate}
\lim_{N\to\infty}\frac{\sum_{k=N-R}^\infty k^\nu f_2(k-1)}{f_3(N)}=0.
\end{align}

Then $\caD_{f_2}\subset D(\delta_{\Phi(s)})\cap  D(\delta_{\dot\Phi(s)})$, and
%
%
there is a constant $C^{(1)}_{f_2,f_3}>0$
such that 
\begin{align}
&\sup_{s\in[0,1]}\lV
\delta_{\Phi(s)}\lmk A\rmk
\rV_{f_3},\;\sup_{N\in\nan}\sup_{s\in[0,1]}\lV
\delta_{\Phi_{N}(s)}\lmk A\rmk
\rV_{f_3}
\le
C^{(1)}_{f_2,f_3}\lV A\rV_{f_2}\label{deltaps}\\
&\sup_{s\in[0,1]}\lV
\delta_{\dot{\Phi}(s)}\lmk A\rmk
\rV_{f_3},\;
\sup_{N\in\nan}\sup_{s\in[0,1]}\lV
\delta_{\dot{\Phi}_{N}(s)}\lmk A\rmk
\rV_{f_3}
\le
C^{(1)}_{f_2,f_3}\lV A\rV_{f_2}\label{deldadps}
\end{align}
for all $A\in\caD_{f_2}$, and $\varepsilon>0$.
If we assume Assumption \ref{assump} (iv) in addition, then we 
may also take $C^{(1)}_{f_2,f_3}>0$ so that 
\begin{align}
&\sup_{s,s_0 \in[0,1],0<| s-s_0|\le \varepsilon}\lV
\delta_{\frac{\Phi(s)-\Phi(s_0)}{s-s_0}-\dot{\Phi}(s_0)}
\lmk A\rmk
\rV_{f_3},\;
\sup_{N\in\nan}\sup_{s,s_0 \in[0,1],0<| s-s_0|\le \varepsilon}\lV
\delta_{\frac{\Phi_{N}(s)-\Phi_{N}(s_0)}{s-s_0}-\dot{\Phi}_{N}(s_0)}
\lmk A\rmk
\rV_{f_3}\nonumber\\
&\le
b(\varepsilon)C^{(1)}_{f_2,f_3}\lV A\rV_{f_2}.\label{dddbs}
\end{align}

\end{lem}
\begin{proof}
We prove (\ref{deltaps}). The proof of (\ref{deldadps}) and (\ref{dddbs}) are same.
Note that there exists a constant $C_5>0$ such that 
\begin{align}
\lV
(H_{\Phi(s)})_{\Lambda_{N+R}}
\rV
\le C_5 \lv \Lambda_{N+R}\rv,\quad s\in[0,1],\quad N\in\nan.
\end{align}
Therefore ,we have
\begin{align}
\lV
\delta_{\Phi(s)}(A_N)
\rV
=\lV
\lcm (H_{\Phi(s)})_{\Lambda_{N+R}}, A_N\rcm
\rV
\le 2C_5\lv \Lambda_{N+R}\rv\lV A_N\rV,\quad A_N\in\caA_{\Lambda_N},\; s\in[0,1].
\end{align}
From this, for any $A\in\caD_{f_2}$ and $N,M\in\nan$ with $M>N$, we have
\begin{align}
&\lV
\delta_{\Phi(s)}
\lmk
\bbE_N(A)-\bbE_M(A)
\rmk
\rV
=\lV
\sum_{k=N+1}^M \delta_{\Phi(s)}
\lmk
\bbE_k(A)-\bbE_{k-1}(A)
\rmk
\rV
\le 2C_5
\sum_{k=N+1}^M\lv \Lambda_{k+R}\rv
\lV
 \bbE_k(A)-\bbE_{k-1}(A)
\rV\notag \\
&\le
4C_5\lV A\rV_{f_2}
\sum_{k=N+1}^M\lv \Lambda_{k+R}\rv
f_2(k-1).
\label{dnm}
\end{align}
Hence $\{
\delta_{\Phi(s)}
\lmk
\bbE_N(A)
\rmk\}_N
$ with $A\in \caD_{f_{2}}$
is a Cauchy sequence in $\caA$, hence there exists a limit $\lim_{N\to\infty}\delta_{\Phi(s)}
\lmk
\bbE_N(A)
\rmk$.
On the other hand, $\bbE_N(A)$ converges to $A$ in $\lV \cdot\rV$.
By the closedness of $\delta_{\Phi(s)}$, $A\in \caD_{f_{2}}$ belongs to the domain
 $D(\delta_{\Phi(s)})$ of $\delta_{\Phi(s)}$, and
 \begin{align}
 \delta_{\Phi(s)}\lmk
 A\rmk=
 \lim_{N\to\infty}\delta_{\Phi(s)}
\lmk
\bbE_N(A)
\rmk.
 \end{align}
Hence we get
$\caD_{f_2}\subset D(\delta_{\Phi(s)})$.
From (\ref{dnm}), we have
\begin{align}
	\begin{split}
&\lV\delta_{\Phi(s)}\lmk A\rmk\rV
=\lim_{N\to\infty}\lV
\delta_{\Phi(s)}
\lmk
\bbE_N(A)
\rmk\rV
=\lim_{N\to\infty}\lV
\delta_{\Phi(s)}
\lmk
\bbE_N(A)-\bbE_1(A)+\bbE_1(A)
\rmk
\rV
\\
&\le
4C_5\lV A\rV_{f_2}
\sum_{k=2}^\infty\lv \Lambda_{k+R}\rv
f_2(k-1)
+ 2C_5\lv \Lambda_{1+R}\rv\lV A\rV\\
&\le
\lmk
4C_5
\sum_{k=2}^\infty\lv \Lambda_{k+R}\rv
f_2(k-1)
+ 2C_5\lv \Lambda_{1+R}\rv
\rmk
\lV A\rV_{f_2},
	\end{split}
\end{align}
for any $A\in\caD_{f_2}$.

Next note that 
\begin{align}
	\begin{split}
&\lV
\delta_{\Phi(s)}(A)-\bbE_N\lmk \delta_{\Phi(s)}(A)\rmk
\rV
=\lim_{M\to\infty}
\lV
\delta_{\Phi(s)}\lmk\bbE_M(A)
\rmk
-\bbE_N\lmk \delta_{\Phi(s)}\lmk \bbE_M(A)\rmk\rmk
\rV\\
&=\lim_{M\to\infty}
\lV
\delta_{\Phi(s)}\lmk\bbE_M(A)-\bbE_{N-R}(A)+\bbE_{N-R}(A)
\rmk
-\bbE_N\lmk \delta_{\Phi(s)}\lmk \bbE_M(A)-\bbE_{N-R}(A)+\bbE_{N-R}(A)\rmk\rmk
\rV\\
&=\lim_{M\to\infty}
\lV
\delta_{\Phi(s)}\lmk\bbE_M(A)-\bbE_{N-R}(A)
\rmk
-\bbE_N\lmk \delta_{\Phi(s)}\lmk \bbE_M(A)-\bbE_{N-R}(A)\rmk\rmk
\rV\\
&\le
8C_5\lV A\rV_{f_2}
\sum_{k=N-R+1}^\infty\lv \Lambda_{k+R}\rv
f_2(k-1),
	\end{split}
\end{align}
for any $A\in\caD_{f_2}$.
Here, in the third line we used the fact that $ \delta_{\Phi(s)}\lmk \bbE_{N-R}(A)\rmk\in\caA_{\Lambda_N}$.
In the fourth line, we used (\ref{dnm}).
Therefore, 
we obtain
\begin{align}
\lV
\delta_{\Phi(s)}(A)
\rV_{f_3}
\le
\lmk
8C_5
\sup_{N\in\nan}\frac{
\sum_{k=N-R+1}^\infty\lv \Lambda_{k+R}\rv
f_2(k-1)}{f_3(N)}
+
4C_5
\sum_{k=2}^\infty\lv \Lambda_{k+R}\rv
f_2(k-1)
+ 2C_5\lv \Lambda_{1+R}\rv
\rmk
\lV A\rV_{f_2}.
\end{align}
The right hand side is finite from the assumptions.
Hence we have shown (\ref{deltaps}).
\end{proof}

\begin{lem}\label{deltalem}Let $f,f_3: (0, \infty) \to (0,\infty)$ be continuous decreasing functions
with $\lim_{t\to\infty} f(t)=\lim_{t\to\infty}f_3(t)=0$
such that
 \begin{align}
&{\sum_{k=1}^\infty k^{\nu} \sqrt{f(k-1)}}<\infty,\\
&\lim_{N\to\infty} \frac{\sum_{k=N-R}^\infty k^{\nu} \sqrt{f(k-1)}}{\lmk f_3(N)\rmk^2} 
=0.
\end{align}

Then we have $\caD_f\subset \caD\lmk \delta_{\dot\Phi(s)}\rmk$,
$\delta_{\dot\Phi(s)}\lmk \caD_{f}\rmk\subset \caD_{f_3}$and
\begin{align}
&\lim_{N\to\infty}
\sup_{s\in[0,1]}\lV
\lmk
\delta_{\dot{\Phi}(s),N}-\delta_{\dot{\Phi}(s)}
\rmk
(A)
\rV_{f_3}
=0,\quad A\in \caD_{f}.
\end{align}
In particular, for each $A\in\caD_{f}$, $[0,1]\ni s\to \delta_{\dot{\Phi}(s)}(A)\in \caD_{f_3}$ is 
continuous with respect to the norm $\lV \cdot\rV_{f_3}$.
The same statement, with $\delta_{\dot{\Phi}(s)}$ replaced by $\delta_{{\Phi}(s)}$ also holds.
\end{lem}
\begin{proof}
We prove the claim for  $\delta_{\dot{\Phi}(s)}$. The proof for $\delta_{{\Phi}(s)}$ is the same.
Set $f_2(t):=\sqrt{f(t)}$ and $f_4(t):=(f_3(t))^2$.As we have $\sup_{N}\frac{f(N)}{f_2(N)}<\infty$, 
$\sup_{N}\frac{f_4(N)}{f_3(N)}<\infty$ we have
$\caD_f\subset \caD_{f_2}$ and $\caD_{f_4}\subset \caD_{f_3}$.
From Lemma \ref{f2f3} with $(f_2,f_3)$ replaced by $(f_2=\sqrt{f},f_4=f_3^2)$,
 we have $\caD_f\subset\caD_{f_2}\subset D(\delta_{\dot\Phi(s)})$,
 and 
$\delta_{\dot\Phi(s)}(\caD_f)\subset \delta_{\dot\Phi(s)}(\caD_{f_2})\subset\caD_{f_4}\subset  \caD_{f_3}$. 
From Lemma \ref{f2f3},  with $(f_2,f_3)$ replaced by $(f_2,f_4)$ for $N>R$, we have
\begin{align}
&\lV
\lmk
\delta_{\dot{\Phi}(s),N}-\delta_{\dot{\Phi}(s)}
\rmk
(A)
\rV
\le
\lV
\lmk
\delta_{\dot{\Phi}(s),N}-\delta_{\dot{\Phi}(s)}
\rmk
\lmk
A-\bbE_{N-R}(A)
\rmk
\rV
+\lV
\lmk
\delta_{\dot{\Phi}(s),N}-\delta_{\dot{\Phi}(s)}
\rmk
\bbE_{N-R}(A)
\rV\notag \\
&=\lV
\lmk
\delta_{\dot{\Phi}(s),N}-\delta_{\dot{\Phi}(s)}
\rmk
\lmk
A-\bbE_{N-R}(A)
\rmk
\rV\le
2C_{f_2f_4}^{(1)}\lV
A-\bbE_{N-R}(A)
\rV_{f_2}\notag \\
&=2C_{f_2f_4}^{(1)}
\lmk
\begin{gathered}
\lV A-\bbE_{N-R}(A)\rV\\
+\sup_{M\in\bbN}\frac
{\lV
 A-\bbE_{N-R}(A)-\bbE_{M}\lmk  A-\bbE_{N-R}(A)\rmk
\rV}
{f_2(M)}
\end{gathered}
\rmk\notag \\
&\le
2C_{f_2f_4}^{(1)}
\lmk
\begin{gathered}
\lV A-\bbE_{N-R}(A)\rV\\
+
\max\left\{
\begin{gathered}
\sup_{N-R\le M\in\bbN}\frac
{\lV
 A-\bbE_{M}\lmk  A\rmk
\rV}
{f_2(M)},\\
\sup_{N-R>M\in\bbN}\frac
{\lV
 A-\bbE_{N-R}(A)
\rV}
{f_2(N-R)}
\end{gathered}
\right\}
\end{gathered}
\rmk\notag \\
&\le
2C_{f_2f_4}^{(1)}
\lmk
\begin{gathered}
f(N-R)\lV A\rV_{f}\\
+
\max\left\{
\begin{gathered}
\sup_{N-R\le M\in\bbN}
{\lV
 A
\rV_{f}}
\frac{f(M)}{f_2(M)},\\
\sup_{N-R>M\in\bbN}
{\lV
 A
\rV_f}
\frac{f(N-R)}{f_2(N-R)}
\end{gathered}
\right\}
\end{gathered}
\rmk\notag \\
&=2C_{f_2f_4}^{(1)}
\lmk
f(N-R)+\sup_{N-R\le L}\lmk \frac{f(L)}{f_2(L)}\rmk
\rmk\lV A\rV_{f}
=2C_{f_2f_4}^{(1)}
\lmk
f(N-R)+\sqrt{f(N-R)}
\rmk\lV A\rV_{f}.
\label{nr}
\end{align}
Here $C_{f_2f_4}^{(1)}$ is a constant independent of $N,s$.
Therefore, we have 
\begin{align}
\lim_{N\to\infty }\sup_{s\in[0,1]}\lV
\lmk
\delta_{\dot{\Phi}(s),N}-\delta_{\dot{\Phi}(s)}
\rmk
(A)
\rV=0,\quad A\in\caD_{f}.
\end{align}
Furthermore, for $A\in\caD_f$, we have
\begin{align}
&\frac{\lV
\lmk
\delta_{\dot{\Phi}(s),N}-\delta_{\dot{\Phi}(s)}
\rmk
(A)
-\bbE_{M}
\lmk\lmk
\delta_{\dot{\Phi}(s),N}-\delta_{\dot{\Phi}(s)}
\rmk(A)\rmk
\rV}{f_3(M)}\notag \\
&\le\left\{
\begin{gathered}
\frac{f_4(M)}{f_3(M)}\lmk
\lV \delta_{\dot{\Phi}(s),N}(A)\rV_{f_4}+\lV \delta_{\dot{\Phi}(s)}(A)\rV_{f_4}
\rmk
\le
2C_{f_2f_4}^{(1)}\frac{f_4(M)}{f_3(M)}\lV A\rV_{f_2}
\le
2{f_3(N-R)} C_{f_2f_4}^{(1)}\lV A\rV_{f_2}
,\\ \text{for}\; M>N-R,\\
\frac{4C_{f_2f_4}^{(1)}\lmk f(N-R)+\sqrt{f(N-R)}\rmk}{f_3(M)}\lV A\rV_{f}
\le \frac{4C_{f_2f_4}^{(1)}\lmk f(N-R)+\sqrt{f(N-R)}\rmk}{f_3(N-R)}\lV A\rV_{f}, \\
\text{for}\;M\le N-R.
\end{gathered}
\right. 
\label{mr}
\end{align}
For $M>N-R$, we used Lemma \ref{f2f3},  with $(f_2,f_3)$ replaced by $(f_2,f_4)$.
 For $M\le N-R$, we used (\ref{nr}).
 As
 \begin{align}
\lim_{N\to\infty} 2{f_3(N-R)} C_{f_2f_4}^{(1)}\lV A\rV_{f_2}
=\lim_{N\to\infty}\frac{4C_{f_2f_4}^{(1)}\lmk f(N-R)+\sqrt{f(N-R)}\rmk}{f_3(N-R)}\lV A\rV_{f}
=0,
 \end{align}
 we get
 \begin{align}\lim_{N\to\infty}
 \sup_{M\in\nan}\lmk\frac{\lV
\lmk
\delta_{\dot{\Phi}(s),N}-\delta_{\dot{\Phi}(s)}
\rmk
(A)
-\bbE_{M}
\lmk\lmk
\delta_{\dot{\Phi}(s),N}-\delta_{\dot{\Phi}(s)}
\rmk(A)\rmk
\rV}{f_3(M)}\rmk
=0,\quad A\in\caD_f.
 \end{align}
From this and (\ref{nr}), 
we have shown the claim of the Lemma.
\end{proof}

\subsection{Proof of Lemma \ref{pre}}
Below, we use the following facts repeatedly: for any $0<\beta<\beta'\le 1$, $0<c,c'$, 
$0<a,a'$, $s\in\bbR$, $l\in\nan$, $r=0,1$, 
and $k\in\bbZ$, we have
\begin{align}
&\lim_{t\to\infty} \frac{t^ke^{-\hat h\lmk t-s\rmk}}{ e^{-t^{\beta}}}=
\lim_{t\to\infty} \frac{t^ke^{-\hat h\lmk \lcm \frac{t}{2}\rcm\rmk}}{ e^{-t^{\beta}}}=
\lim_{t\to\infty} \frac{t^ke^{-\hat h\lmk t- \lcm \frac{t}{2}\rcm\rmk}}{ e^{-t^{\beta}}}=
0,\label{trivial1}\\
&\lim_{t\to\infty}\frac{e^{-t^\beta}}{e^{-\lmk \frac t2\rmk^{\beta}}}
=0,\label{trivial2}\\
&\lim_{t\to\infty}\frac{t^ke^{-t^{\beta'}}}{ e^{-(t)^{\beta}}}=\lim_{t\to\infty}\frac{t^ke^{-c\lmk \lcm\frac t 2\rcm\rmk^{\beta'}}}
{ e^{-t^{\beta}}}=
\lim_{t\to\infty}\frac{t^ke^{-\lmk t-\lcm\frac t 2\rcm\rmk^{\beta'}}}
{ e^{-t^{\beta}}}=0,\label{trivial3}\\
&\sum_{m=1}^\infty m^k e^{-c (m-r)^{\beta}}<\infty,\label{trivial4}\\
&\lim_{N\to\infty}\frac{\sum_{m=N-l}^\infty m^ke^{-c\lmk m-r\rmk^{\beta'}}}{e^{-c'N^\beta}}
\le\sum_{m=1}^\infty m^ke^{-\frac c2 \lmk m-r\rmk ^{\beta'}} \lim_{N\to\infty}
\frac{e^{-\frac c2 (N-l-r)^{\beta'}}}{e^{-c'N^\beta}}=0.
\label{trivial5}
\end{align}
We also note that for $0<\beta<1$, $0<c,c'$, and $l\in\nan$,
$|t|^le^{-\hat h(ct)}/e^{-\lmk c't\rmk^\beta}$ is integrable with respect to $t>0$.
From this and (\ref{omegabound}), for any $0<\beta<1$, $0<c$, and $l\in\nan$, we have
\begin{align}\label{oibd}
\int_{-\infty}^\infty dt \omega_\gamma\lmk t \rmk |t|^le^{\lmk c|t|\rmk^\beta}<\infty.
\end{align}
We also have for any $0<\beta<1$ and $c>0$
\begin{align}\label{Web}
\sup_{t\ge 1}\frac{W_\gamma(c\lcm\frac t2 \rcm)}{e^{-t^\beta}}<\infty,
\end{align}
from (\ref{Omegabound}).

\begin{lem}\label{416lem}
Fix $0<\beta_5<\beta_1<1$ and set
$f(t):=\frac{\exp(-t^{\beta_1})}t$, and $\zeta(t):=\exp(-t^{\beta_5})$.
Then for any $A\in \caD_f$, and $(s',u',s'',s''')\in [0,1]\times\bbR\times [0,1]\times[0,1]$,
we have
$\tau_{\Phi(s'')}^{-u'}\circ \alpha_{s'''}^{-1}(A)\in\caD_{f_{2}}\subset \caD_{\zeta}\subset  D(\delta_{\Phi(s')})\cap D(\delta_{\dot\Phi(s')})$
and $\delta_{{\Phi}(s')}\circ
\tau_{\Phi(s'')}^{-u'}\circ \alpha_{s'''}^{-1}(A),
\delta_{\dot{\Phi}(s')}\circ
\tau_{\Phi(s'')}^{-u'}\circ \alpha_{s'''}^{-1}(A)\in\caD_{\zeta}$.
For any $A\in \caD_f$ and any compact intervals $[a,b]$, $[c,d]$ of $\bbR$, the maps
\begin{align}\label{391a}
[a,b]\times [0,1]\times [0,1]\times[c,d]\times [0,1]\times[0,1]
\ni (u,s,s',u',s'',s''')\mapsto
 \tau_{\Phi(s)}^{u}\circ\delta_{{\Phi}(s')}\circ
\tau_{\Phi(s'')}^{-u'}\circ \alpha_{s'''}^{-1}(A)
\in\caA
\end{align}
and
\begin{align}\label{391}
[a,b]\times [0,1]\times [0,1]\times[c,d]\times [0,1]\times[0,1]
\ni (u,s,s',u',s'',s''')\mapsto
 \tau_{\Phi(s)}^{u}\circ\delta_{\dot{\Phi}(s')}\circ
\tau_{\Phi(s'')}^{-u'}\circ \alpha_{s'''}^{-1}(A)
\in\caA
\end{align}
are uniformly continuous with respect to $\lV\cdot\rV$, and the maps
\begin{align}\label{392a}
 [0,1]\times[c,d]\times [0,1]\times[0,1]
\ni (s',u',s'',s''')\mapsto
\delta_{{\Phi}(s')}\circ
\tau_{\Phi(s'')}^{-u'}\circ \alpha_{s'''}^{-1}(A)
\in\caD_{\zeta}
\end{align}
and 
\begin{align}\label{392}
 [0,1]\times[c,d]\times [0,1]\times[0,1]
\ni (s',u',s'',s''')\mapsto
\delta_{\dot{\Phi}(s')}\circ
\tau_{\Phi(s'')}^{-u'}\circ \alpha_{s'''}^{-1}(A)
\in\caD_{\zeta}
\end{align}
are uniformly continuous with respect to $\lV\cdot\rV_\zeta$.
For any $A\in \caD_f$, the integral 
\begin{align}\label{vswe}
\int dt \omega_\gamma(t) \int_0^t du
\tau_{\Phi(s)}^u\circ\delta_{\dot{\Phi}(s)}\lmk
\tau_{\Phi(s)}^{-u}(A)\rmk,
\end{align}
and 
\begin{align}\label{vswea}
\int dt \;\omega_\gamma(t)\int_0^t du\;
\tau_{\Phi(s)}^{t-u}\circ \lmk \delta_{\dot\Phi(s)}\rmk\circ
\tau_{\Phi(s)}^u(A),
\end{align}
are well-defined as 
Bochner integrals of 
($\caA$, $\lV \cdot \rV$).
{Furthermore, for any $A\in \caD_f$,
$\alpha_s^{-1}(A)$ and $\alpha_s(A)$ are differentiable with respect to
$\lV\cdot\rV$ and
\begin{align}\label{alpha3}
\frac{d}{ds}\alpha_s^{-1}(A)
=\int dt \omega_\gamma(t) \int_0^t du
\tau_{\Phi(s)}^u\circ\delta_{\dot{\Phi}(s)}\lmk
\tau_{\Phi(s)}^{-u}\lmk \alpha_{s}^{-1}(A)\rmk
\rmk.
\end{align}}
{ The right hand side can be understood as a Bochner integral of 
($\caA$, $\lV \cdot \rV$) and there is a constant $C_{9,f}>0$ such that
\begin{align}\label{alphabd}
\lV
\frac{d}{ds} \alpha_{s}^{-1}(A)
\rV, \lV
\frac{d}{ds} \alpha_{s}(A)
\rV
\le C_{9,f}
\lV A\rV_f,\quad A\in\caD_f.
\end{align}
}
\end{lem}
{
\begin{rem}\label{imply}
As mentioned in the introduction, $\alpha_s$ is the same automorphism 
given in \cite{BMNS} and \cite{NSY}.
In particular, if a $C^1$-path of interactions satisfy {\it Condition B} in \cite{tri}
except for the time reversal condition {(iii) {\it 6}}, 
for each $s\in[0,1]$, the unique ground state $\varphi_s$
is given by $\varphi_s=\varphi_0\circ \alpha_s$, with the $\alpha_s$.
Lemma \ref{416lem} implies for any $A\in{\caD}_f$,
$\varphi_s(A)=\varphi_0\circ \alpha_s(A)$ is differentiable 
and the derivative is bounded by 
$ C_{9,f}
\lV A\rV_f$, corresponding to Assumption \ref{assump} (vii).
It is well known that the local gap implies the existence of the gap in the bulk in the sense of Assumption 1.2 (vi), \cite{Ogata3}.

\end{rem}
}

\begin{proof}
We prove the continuity for (\ref{391}) and (\ref{392}). The proof for (\ref{391a}) and (\ref{392a})
are the same.
We also prove only (\ref{vswe}). The proof for (\ref{vswea}) is the same.
{ We prove (\ref{alphabd}) only for $\alpha_s^{-1}$. The proof
for $\alpha_s$ is analogous.}

Choose real numbers $\beta_4,\beta_3,\beta_2$ so that $0<\beta_5<\beta_4<\beta_3<\beta_2<\beta_1<1$ and fix.
Define
$f_0(t):=\exp(-t^{\beta_1})$, $f_1(t):=\exp(-t^{\beta_2})$, $f_2(t):=t^{-2(\nu+2)}\exp(-t^{\beta_3})$, $g(t):=\exp(-t^{\beta_4})$.

Note that
$f_1,f,f_0: (0, \infty) \to (0,\infty)$ are continuous decreasing functions with 
$\lim_{t\to \infty} f_1(t) =\lim_{t\to \infty} f(t) =\lim_{t\to \infty} f_0(t) =0$. 
From (\ref{trivial1}), 
  we have
\begin{align}
&\lim_{N\to\infty}\lmk \frac{e^{-\hat h\lmk N-M\rmk}}{f_1(N)}\rmk
=0,\; \text{for all}\;M\in\nan,\\
&\sup_{N\in\nan}
\frac{e^{-\hat h({ \lcm \frac N2\rcm })}}{f_1(N)}\lv \Lambda_{N-{ \lcm \frac N2\rcm }}\rv
<\infty.\label{232}
	\end{align}
	Furthermore, from (\ref{trivial3}) and $0<\beta_2<\beta_1<1$, we have
\begin{align}
&\sup_{N\in\nan} \frac{f_0(N-{ \lcm \frac N2\rcm })}{f_1(N)}
<\infty.
\end{align}
We also have
\begin{align}
\lim_{M\to\infty}\frac{f(M)}{f_0(M)}
=\lim_{M\to\infty}\frac 1M=0.
\end{align}
Therefore, from Lemma \ref{418} with $(f,f_0,f_1)$ replaced by
$(f_1,f,f_0)$, we have $\alpha_s^{-1}(\caD_f)\subset \caD_{f_1}$ and 
\begin{align}\label{mo}
\sup_{s\in[0,1]}\lV
\alpha_{s,\Lambda_n}^{-1}(A)-\alpha_{s}^{-1}(A)
\rV_{f_1}\to 0,\quad A\in\caD_{f}.
\end{align}
Therefore, for each $A\in\caD_{f}$, $[0,1]\ni s\to \alpha_{s}^{-1}(A)\in \caD_{f_1}$ is 
continuous with respect to the norm $\lV \cdot\rV_{f_1}$.


Note that 
$f, f_1: (0, \infty) \to (0,\infty)$ are continuous decreasing functions with $\lim_{t\to \infty} f_1(t) =\lim_{t\to \infty} f(t) =0$.
From (\ref{trivial3}), and $0<\beta_2<\beta_1<1$, we have
\begin{align}
\sup_{N\in\nan}\lmk \frac{f\lmk N-{\lcm \frac N2\rcm}\rmk}{f_1(N)}\rmk
<\infty.
\end{align}
From this and (\ref{232}),
Lemma \ref{nn46} with $(f_1,f_2)$ replaced by
$(f,f_1)$ implies the existence of a constant $C_{8, f,f_1}>0$
such that
\begin{align}
\sup_{s\in[0,1]}\lV
\alpha_s^{-1}\lmk A\rmk
\rV_{f_1}
\le
C_{8,f,f_1}
\lV A\rV_{f}
,\quad A\in\caD_f.
\end{align}

The functions $f_1,f_2: (0, \infty) \to (0,\infty)$ are continuous decreasing functions with $\lim_{t\to \infty} f_i(t) =0$,
$i=1,2$.
From (\ref{trivial3}), we have
\begin{equation}
\lim_{N\to\infty}\lmk
 \frac{|\Lambda_{\lcm \frac N2\rcm}|  e^{-(N-\lcm \frac N2\rcm)}}{f_2(N)}\rmk
 =0.
	\end{equation}
	From (\ref{trivial3}) and $0<\beta_3<\beta_2<1$, we have
\begin{align}
\lim_{N\to\infty}\frac{f_1\lmk \lcm \frac N 2\rcm\rmk}{f_2(N)}
=
0.
\end{align}
Therefore, from Lemma \ref{44},
we have
\begin{align}\label{nt}
\sup_{s\in[0,1]}\lV
\tau_{\Phi(s),\Lambda_n}^{-u}(A)
-\tau_{\Phi(s)}^{-u}(A)
\rV_{f_2}
\to 0,\quad A\in \caD_{f_1},
\end{align}
uniformly in compact $u\in \bbR$.
Therefore, for each $A\in\caD_{f_1}$,
$\bbR\times [0,1]\ni (u,s)\to \tau_{\Phi(s)}^{-u}(A)\in \caD_{f_2}$ is 
continuous with respect to the norm $\lV \cdot\rV_{f_2}$.

Note that $f_2,\zeta: (0, \infty) \to (0,\infty)$ are continuous decreasing functions
with $\lim_{t\to\infty}f(t)=\lim_{t\to\infty}\zeta(t)=0$.
From (\ref{trivial4}) and (\ref{trivial5}), and $0<\beta_5<\beta_3<1$, we have
 \begin{align}\label{114}
&{\sum_{k=1}^\infty k^{\nu} \sqrt{f_2(k)}}
<\infty,\\
&\lim_{N\to\infty} \frac{\sum_{k=N-R}^\infty k^{\nu}\sqrt{f_2(k)}}{\zeta(N)^2}=0.
\end{align}
Therefore applying Lemma \ref{deltalem} with $(f,f_3)$ replaced by
$(f_2,\zeta)$, we have 
$\delta_{\dot{\Phi}(s)}(\caD_{f_2})\subset \caD_{\zeta}$ and
\begin{align}\label{el}
\lim_{N\to\infty}
\sup_{s\in[0,1]}\lV
\lmk
\delta_{\dot{\Phi}_N(s)}-\delta_{\dot{\Phi}(s)}
\rmk
(A)
\rV_{\zeta}
=0,\quad A\in \caD_{f_2}.
\end{align}
Therefore, for each $A\in\caD_{f_2}$, $[0,1]\ni s\to \delta_{\dot{\Phi}(s)}(A)\in \caD_\zeta$ is 
continuous with respect to the norm $\lV \cdot\rV_\zeta$.

Note that $f_2:(0,\infty)\to (0,\infty)$
is a continuous decreasing function
with $\lim_{t\to\infty}f_2(t)=0$.
From (\ref{oibd}), we have
\begin{align}
&\int_{(4v\lv t\rv)\ge 1} dt\; \omega_\gamma(t) \frac{|t|}{f_2((4v\lv t\rv))}
<\infty.
\end{align}
We also have 
\begin{align}
&\sup_{N\in\nan}\frac{f_1(N-{ \lcm \frac N2\rcm })}{f_2(N)}
<\infty,\\
&\sup_{N\in\nan}\frac{\lv\Lambda_N\rv e^{-\frac{{ \lcm \frac N2\rcm}}2}}{f_2(N)}
<\infty,
\end{align}
from (\ref{trivial3}) with $0<\beta_3<\beta_2<1$ and
(\ref{trivial1}).
Therefore, from Lemma \ref{48}, 
with $(f,f_1)$ replaced by $(f_1,f_2)$
we have $\tau_{\Phi(s)}^t\lmk\caD_{f_1}\rmk\subset \caD_{f_2}$ and
there is a non-negative non-decreasing function on $\bbR_+$, $b_{1,f_1,f_2}(t)$ such that
\begin{align}\label{b1f1f2}
\int\; dt\; \omega_\gamma(t)\; |t|\cdot b_{1,f_1,f_2}(|t|)<\infty
\end{align}
and 
\begin{align}
&
\sup_{s\in[0,1]}\lV
\tau_{\Phi(s)}^t\lmk
A
\rmk
\rV_{f_2},\sup_{N\in\nan}\sup_{s\in[0,1]}\lV
\tau_{\Phi_N(s)}^t\lmk
A
\rmk
\rV_{f_2}
\le
b_{1,f_1,f_2}(|t|)\lV A\rV_{f_1},\quad A\in\caD_{f_1}.
\end{align}

Note that
 $f_2,\zeta: (0, \infty) \to (0,\infty)$ are continuous decreasing functions
such that
$\lim_{t\to\infty}f_2(t)=\lim_{t\to\infty} \zeta(t)=0$.
By (\ref{trivial4}) and (\ref{trivial5}) with $0<\beta_5<\beta_3<1$, we have
\begin{align}\label{114}
&{\sum_{k=2}^\infty k^\nu f_2(k-1)}<\infty,\\
&\limsup_N
\frac{\sum_{k=N-R}^\infty k^\nu f_2(k)}{\zeta(N)}
=0.
\end{align}
Therefore, from Lemma \ref{f2f3} with $(f_2,f_3)$ replaced by
$(f_2,\zeta)$,  we have 
$\caD_{f_2}\subset D(\delta_{\Phi(s)})\cap  D(\delta_{\dot\Phi(s)})\cap D(\delta_{\frac{\Phi(s)-\Phi(s_0)}{s-s_0}-\dot{\Phi}(s_0)})$, and there exists a constant $C_{7,f_2,\zeta}^{(1)}>0$ such that
\begin{align}
&\sup_{s\in[0,1]}\lV
\delta_{\Phi(s)}\lmk A\rmk
\rV_\zeta,\quad \sup_{N\in\nan}\sup_{s\in[0,1]}\lV
\delta_{\Phi_N(s)}\lmk A\rmk
\rV_\zeta
\le
C_{7,f_2,\zeta}^{(1)}\lV A\rV_{f_2}\label{deltap}\\
&\sup_{s\in[0,1]}\lV
\delta_{\dot{\Phi}(s)}\lmk A\rmk
\rV_\zeta,\quad \sup_{N\in\nan}\sup_{s\in[0,1]}\lV
\delta_{{\dot\Phi_N(s)}}\lmk A\rmk
\rV_\zeta
\le
C_{7,f_2,\zeta}^{(1)}\lV A\rV_{f_2}\label{deldadp}
\end{align}
for all $A\in\caD_{f_2}$ and $\varepsilon>0$.

We claim that for any compact intervals $[a,b]$, $[c,d]$ of $\bbR$
and $A\in\caD_f$,
\begin{align}
[a,b]\times [0,1]\times [0,1]\times[c,d]\times [0,1]\times[0,1]
\ni (u,s,s',u',s'',s''')\mapsto
 \tau_{\Phi(s)}^{u}\circ\delta_{\dot{\Phi}(s')}\circ
\tau_{\Phi(s'')}^{-u'}\circ \alpha_{s'''}^{-1}(A)
\in\caA
\end{align}
is continuous with respect to $\lV\cdot\rV$.
We also claim that
\begin{align}
[0,1]\times[c,d]\times [0,1]\times[0,1]
\ni (s',u',s'',s''')\mapsto
\delta_{\dot{\Phi}(s')}\circ
\tau_{\Phi(s'')}^{-u'}\circ \alpha_{s'''}^{-1}(A)
\in\caD_{\zeta}
\end{align}
is continuous with respect to $\lV\cdot\rV_\zeta$.

To see this, let $A\in\caD_{f}$ and  fix any $\varepsilon>0$.
Note that from the continuity of $[0,1]\ni s'''\mapsto \alpha_{s'''}^{-1}(A)\in \caD_{f_1}$ in $\lV\cdot\rV_{f_1}$,
there exists a finite sequence $s_0=0<s_1<\cdots<s_{N_\varepsilon}=1$
such that
\begin{align}
\lV
\alpha_{s'''}^{-1}(A)-\alpha_{s_i}^{-1}(A)
\rV_{f_1}<\varepsilon,
\text{for all }s'''\in[s_{i-1},s_{i+1}],\text{and } i=1,\ldots,N_{\varepsilon}-1.
\end{align}
For $\alpha_{s_i}^{-1}(A)\in\caD_{f_1}$, $i=0,\ldots, N_\varepsilon$, from
the continuity of $(u',s'')\mapsto \tau_{\Phi(s'')}^{-u'}\circ \alpha_{s_i}^{-1}(A)\in\caD_{f_2}$, in $\lV\cdot\rV_{f_2}$ we get
$\tilde s_0=0<\tilde s_1<\cdots<\tilde s_{\tilde N_\varepsilon}=1$
and $u_0=c<u_1<\cdots<u_{M_\varepsilon}=d$
such that
\begin{align}
\lV
\lmk \tau_{\Phi(s'')}^{-u'}-\tau_{\Phi(\tilde s_j)}^{-u_k}\rmk
\circ \alpha_{s_i}^{-1}(A)
\rV_{f_2}<\varepsilon,\\
\text{for all }s''\in[\tilde s_{j-1},\tilde s_{j+1}],\text{and } j=1,\ldots,\tilde N_{\varepsilon}-1,\nonumber\\
\text{for all }u'\in[u_{k-1}, u_{k+1}],\text{and } k=1,\ldots, M_{\varepsilon}-1,\nonumber\\
\text{and } i=1,\ldots, N_{\varepsilon}-1.
\end{align}
From
the continuity of $[0,1]\ni s'\to\delta_{\dot{\Phi}(s')}\circ
\tau_{\Phi(\tilde s_j)}^{-u_k}\circ \alpha_{s_i}^{-1}(A)\in\caD_\zeta$
for $\tau_{\Phi(\tilde s_j)}^{-u_k}\circ \alpha_{s_i}^{-1}(A)\in\caD_{f_2}$
in $\lV\cdot \rV_\zeta$,
there exists
a finite sequence $\hat s_0=0<\hat s_1<\cdots<\hat s_{\hat N_\varepsilon}=1$
such that
\begin{align}
	\begin{split}
\lV
\lmk \delta_{\dot{\Phi}(s')}-\delta_{\dot{\Phi}(\hat s_l)}\rmk\circ
\tau_{\Phi(\tilde s_j)}^{-u_k}\circ \alpha_{s_i}^{-1}(A)
\rV_\zeta
<\varepsilon.\\
\text{for all }s'\in[\hat s_{l-1},\hat s_{l+1}],\text{and } l=1,\ldots,\hat N_{\varepsilon}-1,\\
 \text{and }j=1,\ldots,\tilde N_{\varepsilon}-1,
\text{and } k=1,\ldots, M_{\varepsilon}-1,\\
\text{and } i=1,\ldots,N_{\varepsilon}-1.
	\end{split}
\end{align}
Finally, from the continuity of 
$\bbR\times [0,1]\ni (u,s)\to \tau_{\Phi(s)}^{u}\lmk
 \delta_{\dot{\Phi}(\hat s_l)}\circ
\tau_{\Phi(\tilde s_j)}^{-u_k}\circ \alpha_{s_i}^{-1}(A)
\rmk\in\caA$ 
in the norm $\lV \cdot\rV$,
( Lemma \ref{ao},)
we have
finite sequences $\check s_0=0<\check s_1<\cdots<\check s_{\check N_\varepsilon}=1$
and $\hat u_0=a<\hat u_1<\cdots<\hat u_{\hat M_\varepsilon}=b$
such that
\begin{align}
	\begin{split}
\lV
\lmk \tau_{\Phi(s)}^{u}-\tau_{\Phi(\check s_y)}^{\hat u_x}\rmk
\circ \delta_{\dot{\Phi}(\hat s_l)}\circ
\tau_{\Phi(\tilde s_j)}^{-u_k}\circ \alpha_{s_i}^{-1}(A)
\rV<\varepsilon,\\
\text{for all }s\in[\check s_{y-1},\check s_{y+1}],\text{and } y=1,\ldots,\check N_{\varepsilon}-1,\\
\text{and }u\in[\hat u_{x-1},\hat u_{x+1}],
\text{and } x=1,\ldots, \hat M_{\varepsilon}-1,\\
\text{and } l =1,\ldots,\hat N_{\varepsilon}-1,\\
 \text{and }j=1,\ldots,\tilde N_{\varepsilon}-1,
\text{and } k=1,\ldots, M_{\varepsilon}-1,\\
\text{and } i=1,\ldots, N_{\varepsilon}-1.
	\end{split}
\end{align}
Now for any $(u,s,s',u',s'',s''')\in [a,b]\times [0,1]\times [0,1]\times[c,d]\times [0,1]\times[0,1]$,
there is $(x,y,l,k,j,i)$
such that
\begin{align}
u\in [\hat u_{x-1},\hat u_{x+1}],  s\in[\check s_{y-1},\check s_{y+1}],
s'\in[\hat s_{l-1},\hat s_{l+1}],
u'\in[ u_{k-1}, u_{k+1}],s''\in[\tilde s_{j-1},\tilde s_{j+1}],
s'''\in[s_{i-1},s_{i+1}].
\end{align}
For any such $(x,y,l,k,j,i)$,
we have
\begin{align}
	\begin{split}
&\lV
-\tau_{\Phi(\check s_y)}^{\hat u_x}
\circ \delta_{\dot{\Phi}(\hat s_l)}\circ
\tau_{\Phi(\tilde s_j)}^{-u_k}\circ \alpha_{s_i}^{-1}(A)
+\tau_{\Phi(s)}^{u}\circ\delta_{\dot{\Phi}(s')}\circ
\tau_{\Phi(s'')}^{-u'}\circ \alpha_{s'''}^{-1}(A)
\rV\\
&\le
\lV
\lmk \tau_{\Phi(s)}^{u}-\tau_{\Phi(\check s_y)}^{\hat u_x}\rmk
\circ \delta_{\dot{\Phi}(\hat s_l)}\circ
\tau_{\Phi(\tilde s_j)}^{-u_k}\circ \alpha_{s_i}^{-1}(A)
\rV
+
\lV
\tau_{\Phi(s)}^{u}\circ \lmk-\delta_{\dot{\Phi}(\hat s_l)}+\delta_{\dot{\Phi}(s')}\rmk\circ
\tau_{\Phi(\tilde s_j)}^{-u_k}\circ \alpha_{s_i}^{-1}(A)
\rV\\
&+
\lV\tau_{\Phi(s)}^{u}
\circ \delta_{\dot{\Phi}(s')}\circ
\lmk -\tau_{\Phi(\tilde s_j)}^{-u_k}+\tau_{\Phi(s'')}^{-u'}\rmk
\circ \alpha_{s_i}^{-1}(A)
\rV
+\lV
\tau_{\Phi(s)}^{u}
\circ \delta_{\dot{\Phi}(s')}\circ
\tau_{\Phi(s'')}^{-u'}\circ 
\lmk -\alpha_{s_i}^{-1}(A)+\alpha_{s'''}^{-1}(A)\rmk
\rV\\
&
\le
2\varepsilon 
+C_{7,f_2,\zeta}^{(1)}\varepsilon+C_{7,f_2,\zeta}^{(1)}\sup_{u\in [c,d]}b_{1,f_1,f_2}(|u|)\varepsilon.
	\end{split}
\end{align}
We also have
\begin{align}
	\begin{split}
&\lV
-\delta_{\dot{\Phi}(\hat s_l)}\circ
\tau_{\Phi(\tilde s_j)}^{-u_k}\circ \alpha_{s_i}^{-1}(A)
+\delta_{\dot{\Phi}(s')}\circ
\tau_{\Phi(s'')}^{-u'}\circ \alpha_{s'''}^{-1}(A)
\rV_\zeta\\
&\le
\lV
\lmk-\delta_{\dot{\Phi}(\hat s_l)}+\delta_{\dot{\Phi}(s')}\rmk\circ
\tau_{\Phi(\tilde s_j)}^{-u_k}\circ \alpha_{s_i}^{-1}(A)
\rV_\zeta\\
&+
\lV
\delta_{\dot{\Phi}(s')}\circ
\lmk -\tau_{\Phi(\tilde s_j)}^{-u_k}+\tau_{\Phi(s'')}^{-u'}\rmk
\circ \alpha_{s_i}^{-1}(A)
\rV_\zeta
+\lV
\delta_{\dot{\Phi}(s')}\circ
\tau_{\Phi(s'')}^{-u'}\circ 
\lmk -\alpha_{s_i}^{-1}(A)+\alpha_{s'''}^{-1}(A)\rmk
\rV_\zeta\\
&
\le
\varepsilon 
+C_{7,f_2,\zeta}^{(1)}\varepsilon+C_{7,f_2,\zeta}^{(1)}\sup_{u\in [c,d]}b_{1,f_1,f_2}(|u|)\varepsilon.
	\end{split}
\end{align}
As $b_{1,f_1,f_2}$ is an $\bbR$-valued nondecreasing function, $\sup_{u\in [c,d]}b_{1,f_1,f_2}(|u|)$
is finite.
Hence we have proven the continuity of (\ref{391}) and (\ref{392}).

%

Furthermore, for any $A\in\caD_f$, we have
\begin{align}
	\begin{split}
&\sup_{s\in[0,1]}\int dt\
 \omega_\gamma(t) \int_{[0,t]} du
\lV\tau_{\Phi(s)}^u\circ\delta_{\dot{\Phi}(s)}\lmk
\tau_{\Phi(s)}^{-u}\lmk \alpha_{s}^{-1}(A)\rmk
\rmk
\rV\\
&\le
\sup_{s\in[0,1]}\int dt\
 \omega_\gamma(t) \int_{[0,t]} du
 C_{7,f_2,\zeta}^{(1)}b_{1, f_1,f_2}(|u|) C_{8,f,f_1}\lV A\rV_f\\
& \le C_{7,f_2,\zeta}^{(1)}C_{8,f,f_1} \lV A\rV_f
\int dt
 \omega_\gamma(t) 
 b_{1, f_1,f_2}(|t|)|t|
<\infty.
	\end{split}
\end{align}
In the last line we used the fact that $b_{1,f_1,f_2}$ is nondecreasing and (\ref{b1f1f2}).
Therefore, the right hand side of (\ref{alpha3}) is a well-defined Bochner integral of $(\caA,\lV \cdot \rV)$  for any $A\in\caD_f$.
By the same argument,  (\ref{vswe}) is a well-defined Bochner integral 
of $(\caA,\lV \cdot \rV)$ for any $A\in\caD_f$.
By the definition of $\alpha_{s,\Lambda_n}$, we have
\begin{align}
	\begin{split}
&\frac{d}{ds}\alpha_{s,\Lambda_n}^{-1}(A)
= i \lcm
D_{\Lambda_n}(s), \alpha_{s,\Lambda_n}^{-1}(A)
\rcm=i\int dt \omega_\gamma(t) \int_0^t du
\lcm 
\tau_{\Phi(s),\Lambda_n}^u\lmk H_{\dot\Phi(s),\Lambda_n}\rmk,
\alpha_{s,\Lambda_n}^{-1}(A)
\rcm\\
&=\int dt \omega_\gamma(t) \int_0^t du
\tau_{\Phi(s),\Lambda_n}^u\circ\delta_{\dot{\Phi}_n(s)}\lmk
\tau_{\Phi(s),\Lambda_n}^{-u}\lmk \alpha_{s,\Lambda_n}^{-1}(A)\rmk
\rmk,\quad A\in\caD_{f}.
	\end{split}
\end{align}
Hence we obtain
\begin{align}\label{276n}
\alpha_{s,\Lambda_n}^{-1}(A)-\alpha_{s_0,\Lambda_n}^{-1}(A)
=\int_{s_0}^{s}dv\int dt \omega_\gamma(t) \int_0^t du
\tau_{\Phi(v),\Lambda_n}^u\circ\delta_{\dot{\Phi}_n(v)}\lmk
\tau_{\Phi(v),\Lambda_n}^{-u}\lmk \alpha_{v,\Lambda_n}^{-1}(A)\rmk
\rmk,\quad A\in\caD_{f}.
\end{align}
For each $(u,v)$, for any $A\in\caD_f$, we have
\begin{align}
	\begin{split}
&\lV
\tau_{\Phi(v),\Lambda_n}^u\circ\delta_{\dot{\Phi}_n(v)}\circ
\tau_{\Phi(v),\Lambda_n}^{-u}\circ \alpha_{v,\Lambda_n}^{-1}(A)
-\tau_{\Phi(v)}^u\circ\delta_{\dot{\Phi}(v)}\circ
\tau_{\Phi(v)}^{-u}\circ \alpha_{v}^{-1}(A)
\rV\\
&\le
\lV
\tau_{\Phi(v),\Lambda_n}^u\circ\delta_{\dot{\Phi}_n(v)}\circ
\tau_{\Phi(v),\Lambda_n}^{-u}\lmk \alpha_{v,\Lambda_n}^{-1}(A)-\alpha_{v}^{-1}(A)
\rmk
\rV
+\lV
\tau_{\Phi(v),\Lambda_n}^u\circ\delta_{\dot{\Phi}_n(v)}\circ
\lmk
\tau_{\Phi(v),\Lambda_n}^{-u}-\tau_{\Phi(v)}^{-u}
\rmk
\alpha_{v}^{-1}(A)
\rV\\
&+\lV
\tau_{\Phi(v),\Lambda_n}^u\circ\lmk \delta_{\dot{\Phi}_n(v)}
-\delta_{\dot{\Phi}(v)}
\rmk\lmk
\tau_{\Phi(v)}^{-u}\circ
\alpha_{v}^{-1}(A)\rmk
\rV+
\lV
\lmk \tau_{\Phi(v),\Lambda_n}^u-\tau_{\Phi(v)}^u
\rmk\circ\lmk \delta_{\dot{\Phi}(v)}
\circ
\tau_{\Phi(v)}^{-u}\circ
\alpha_{v}^{-1}(A)\rmk
\rV\\
&
\le
 C_{7,f_2,\zeta}^{(1)}b_{1, f_1,f_2}(|u|)\lV
 \alpha_{v,\Lambda_n}^{-1}(A)-\alpha_{v}^{-1}(A)
\rV_{f_1}
+C_{7,f_2,\zeta}^{(1)}
\lV
\lmk
\tau_{\Phi(v),\Lambda_n}^{-u}-\tau_{\Phi(v)}^{-u}
\rmk
\alpha_{v}^{-1}(A)
\rV_{f_2}\\
&+\lV
\lmk \delta_{\dot{\Phi}_n(v)}
-\delta_{\dot{\Phi}(v)}
\rmk\lmk
\tau_{\Phi(v)}^{-u}\circ
\alpha_{v}^{-1}(A)\rmk
\rV
+
\lV
\lmk \tau_{\Phi(v),\Lambda_n}^u-\tau_{\Phi(v)}^u
\rmk\circ\lmk \delta_{\dot{\Phi}(v)}
\circ
\tau_{\Phi(v)}^{-u}\circ
\alpha_{v}^{-1}(A)\rmk
\rV.
	\end{split}
\end{align}
From (\ref{mo}), (\ref{nt}), (\ref{el}) and Lemma \ref{ao}, the last part converges to $0$
as $n\to\infty$.
Furthermore, we have
\begin{align}
\sup_{n\in\nan}\lV
\tau_{\Phi(v),\Lambda_n}^u\circ\delta_{\dot{\Phi}_n(v)}\circ
\tau_{\Phi(v),\Lambda_n}^{-u}\circ \alpha_{v,\Lambda_n}^{-1}(A)-\tau_{\Phi(v)}^u\circ\delta_{\dot{\Phi}(v)}\circ
\tau_{\Phi(v)}^{-u}\circ \alpha_{v}^{-1}(A)
\rV
\le
 2C_{7,f_2,\zeta}^{(1)}b_{1, f_1,f_2}(|u|) C_{8,f,f_1}\lV A\rV_f,
\end{align}
with
\begin{align}\label{bbb}
\int_0^1 ds\int dt\
 \omega_\gamma(t) \int_{[0,t]} du
 2C_{7,f_2,\zeta}^{(1)}b_{1, f_1,f_2}(|u|) C_{8,f,f_1}\lV A\rV_f
<\infty.
\end{align}
%
%
Therefore, applying Lebesgue's convergence theorem for (\ref{276n}),
we
obtain
\begin{align}\label{alpha4}
\alpha_{s}^{-1}(A)-\alpha_{s_0}^{-1}(A)
=\int_{s_0}^{s}dv\int dt \omega_\gamma(t) \int_0^t du
\tau_{\Phi(v)}^u\circ\delta_{\dot{\Phi}(v)}\lmk
\tau_{\Phi(v)}^{-u}\lmk \alpha_{v}^{-1}(A)\rmk
\rmk,\quad A\in\caD_f.
\end{align}
From this, for $A\in\caD_f$, we get
\begin{align}
	\begin{split}
&\lV
\frac{\alpha_{s}^{-1}(A)-\alpha_{s_0}^{-1}(A)}{s-s_0}
-\int dt \omega_\gamma(t) \int_0^t du
\tau_{\Phi(s_0)}^u\circ\delta_{\dot{\Phi}(s_0)}\lmk
\tau_{\Phi(s_0)}^{-u}\lmk \alpha_{s_0}^{-1}(A)\rmk
\rmk
\rV\\
&\le
\int dt \omega_\gamma(t) \int_{[0,t]} du
\lV
\frac{1}{s-s_0}
\int_{s_0}^{s}dv\lmk
\tau_{\Phi(v)}^u\circ\delta_{\dot{\Phi}(v)}\lmk
\tau_{\Phi(v)}^{-u}\lmk \alpha_{v}^{-1}(A)\rmk
\rmk
-\tau_{\Phi(s_0)}^u\circ\delta_{\dot{\Phi}(s_0)}\lmk
\tau_{\Phi(s_0)}^{-u}\lmk \alpha_{s_0}^{-1}(A)\rmk
\rmk\rmk
\rV.
	\end{split}
\end{align}
By the continuity of 
$(s,u)\to \tau_{\Phi(s)}^u\circ\delta_{\dot{\Phi}(s)}\lmk
\tau_{\Phi(s)}^{-u}\lmk \alpha_{s}^{-1}(A)\rmk
\rmk\in\caA$ with respect to $\lV\cdot\rV$
for $A\in\caD_f$,
we have
\begin{align}
\lim_{s\to s_0}\frac{1}{s-s_0}
\int_{s_0}^{s}dv\lmk
\tau_{\Phi(v)}^u\circ\delta_{\dot{\Phi}(v)}\lmk
\tau_{\Phi(v)}^{-u}\lmk \alpha_{v}^{-1}(A)\rmk
\rmk
-\tau_{\Phi(s_0)}^u\circ\delta_{\dot{\Phi}(s_0)}\lmk
\tau_{\Phi(s_0)}^{-u}\lmk \alpha_{s_0}^{-1}(A)\rmk
\rmk\rmk
=0,
\end{align}
for each $u$.
On the other hand, we have
\begin{align}
	\begin{split}
&\lV
\frac{1}{s-s_0}
\int_{s_0}^{s}dv\lmk
\tau_{\Phi(v)}^u\circ\delta_{\dot{\Phi}(v)}\lmk
\tau_{\Phi(v)}^{-u}\lmk \alpha_{v}^{-1}(A)\rmk
\rmk
-\tau_{\Phi(s_0)}^u\circ\delta_{\dot{\Phi}(s_0)}\lmk
\tau_{\Phi(s_0)}^{-u}\lmk \alpha_{s_0}^{-1}(A)\rmk
\rmk\rmk
\rV\\
&\le
2C_{7,f_2,\zeta}^{(1)}b_{1, f_1,f_2}(|u|) C_{8,f,f_1}\lV A\rV_f,
	\end{split}
\end{align}
with (\ref{bbb}).
From Lebesgue's convergence theorem, we obtain
\begin{align}
\lim_{s\to s_0}\lV
\frac{\alpha_{s}^{-1}(A)-\alpha_{s_0}^{-1}(A)}{s-s_0}
-\int dt \omega_\gamma(t) \int_0^t du
\tau_{\Phi(s_0)}^u\circ\delta_{\dot{\Phi}(s_0)}\lmk
\tau_{\Phi(s_0)}^{-u}\lmk \alpha_{s_0}^{-1}(A)\rmk
\rmk
\rV=0,\quad \text{for}\; A\in\caD_f.
\end{align}
Hence for $A\in\caD_f$,
$[0,1]\ni s\mapsto \alpha_{s}^{-1}(A)$ is differentiable with respect to $\lV\cdot\rV$,
and we have
\begin{align}
\frac{d}{ds} \alpha_{s}^{-1}(A)
=\int dt \omega_\gamma(t) \int_0^t du
\tau_{\Phi(s)}^u\circ\delta_{\dot{\Phi}(s)}
\circ
\tau_{\Phi(s)}^{-u}\lmk \alpha_{s}^{-1}(A)\rmk.
\end{align}
From this formula, we obtain
\begin{align}
\lV
\frac{d}{ds} \alpha_{s}^{-1}(A)
\rV
\le\lmk
\int dt \omega_\gamma(t) \int_{[0,t]}
C_{7,f_2,\zeta}^{(1)}b_{1, f_1,f_2}(|u|) C_{8,f,f_1}
\rmk
\lV A\rV_f=:C_{9,f}\lV A\rV_f.
\end{align}

\end{proof}

Now we prove Lemma \ref{pre}.
\begin{proofof}[Lemma \ref{pre}] $~$

\medskip

1. The inclusions $\mathcal{D}_f\subset \caD_{f_0} \subset \mathcal{D}_{f_1} \subset \mathcal{D}_{f_2} \subset \mathcal{D}_g \subset \mathcal{D}_{\zeta}$ follow by the monotone choice of the $\beta_i$, $i=1,\ldots, 5$.
From (\ref{trivial1}), we can see that
$f$ satisfies the condition required in  Lemma \ref{42}.
Therefore, from Lemma \ref{42}, we have  $\alpha_s^{-1} ( \mathcal{A}_{loc}) \subset \mathcal{D}_f$ for all $s\in [0,1]$. 

\medskip 

2. This is from Lemma \ref{48}.
From (\ref{oibd}), (\ref{trivial3}) $(f,f_1)$ satisfies the conditions required in Lemma \ref{48}.

\medskip

3. Fix $0<\beta_6<\beta_5$ and set $\zeta_0(t):=e^{-t^{\beta_6}}$ for $t>0$.
We apply Lemma \ref{f2f3}, replacing $(f_2,f_3)$ in it by $(\zeta,\zeta_0)$.
To see that $(\zeta,\zeta_0)$ satisfy the required conditions in Lemma \ref{f2f3},
we recall (\ref{trivial4}) and (\ref{trivial5}).
Hence from Lemma \ref{f2f3}, we obtain  $\mathcal{D}_\zeta \subset D( \delta_{\Phi(s)}) \cap D(\delta_{\dot{\Phi}(s)})$. 

\medskip 

4. This also follows by Lemma  \ref{f2f3} with
$(f_2,f_3)$ replaced by $(f_2,\zeta)$.
The required conditions in Lemma  \ref{f2f3} can be checked by
(\ref{trivial4}) and (\ref{trivial5}).

\medskip

5., 6., and 7. are proven in Lemma \ref{416lem}.

\medskip 

8.
This follows from Lemma \ref{410} for $(f,f_1)$.
The conditions for $(f,f_1)$ can be checked from 
(\ref{trivial3}) and (\ref{Web}).
\medskip

9. This is Lemma \ref{ao}.
\medskip

10.
For any $A\in\caD_f$, from {\it 5.} above, 
$(u,s)\mapsto \delta_{\Phi(s)}\circ\tau_{\Phi(s)}^{u}(A)\in\caD_\zeta$ is continuous with respect to
$\lV \cdot\rV_\zeta$.
Furthermore, from {\it 4., 2.,}  above, as in (\ref{ttta}), we have
\begin{align}
	\begin{split}
&\lV \delta_{\Phi(s)}\circ\tau_{\Phi(s)}^{u}(A)\rV_\zeta
\le
C_{f_2,\zeta}^{(1)} \lV
\tau_{\Phi(s)}^{u}(A)
\rV_{f_2}
\le
C_{f_2,\zeta}^{(1)} 
\lmk
1+\sup_{N\in\nan}\frac{f_1(N)}{f_2(N)}\rmk
\lV
\tau_{\Phi(s)}^{u}(A)
\rV_{f_1}\\
&\le
C_{f_2,\zeta}^{(1)}b_{f,f_1}(|u|)
\lmk 1+\sup_N\frac{f_1(N)}{f_2(N)}\rmk\lV A\rV_{f}.
	\end{split}
\end{align}
From {\it 2.} above, the inequality (\ref{b11}) holds and
(\ref{zetabd})
 is well-defined as 
the Bochner integral with respect to
$(\caD_{\zeta}, \lV \cdot\rV_\zeta)$.

\end{proofof}

{\bf Acknowledgment.}\\
{
Y.O. is grateful to  Wojciech De Roeck, Martin Fraas, and Hal Tasaki for fruitful discussion.
A discussion with Hal Tasaki was the starting point of this project.
A.M. is supported in part by National Science Foundation Grant DMS 1813149.
Y.O. is supported by JSPS KAKENHI Grant Number 16K05171 and 19K03534. 
Part of this work was done during the visit of the authors to
 CRM, 
 with the support of CRM-Simons program “Mathematical challenges in many-body physics 
 and quantum information”.
 
}
\bigskip

\appendix
\section{Conditional expectation $\bbE_N$}
We now briefly describe a family of conditional expectations $\{ \mathbb{E}_N : \mathcal{A} \to \mathcal{A}_{ \Lambda_N} ~|~ N\in \mathbb{N} \}$ are used extensively in this paper. Let $N \in \mathbb{N}$ be fixed and let $\Lambda$ denote any finite set containing $\Lambda_N$. Define:
	\begin{equation}\label{eq:cond-exp}
		\begin{split}
\mathbb{E}^\Lambda_N = \mathrm{id} _{ \Lambda_N} \otimes \rho _{\Lambda \setminus \Lambda_N}
		\end{split}
	\end{equation} 
where $\rho _X$ is the product state whose factors are normalized trace: 
	\begin{equation}
		\begin{split}
\rho _X = \frac{1}{d^{|X|}} \bigotimes _{x\in X} \mathrm{tr} _x . 
		\end{split}
	\end{equation}
Each $\mathbb{E}^\Lambda_N$ is bounded and linear, and as $\Lambda \subset \Sigma$ implies $ \mathbb{E} ^\Sigma _N | _{\mathcal{A}_\Lambda} = \mathbb{E} ^ \Lambda _N$, there exists a unique bounded map and conditional expectation $\mathbb{E} _N : \mathcal{A} \to \mathcal{A}_{\Lambda_N} $ such that for all $\Lambda$ containing $\Lambda_N$:
	\begin{equation}
		\begin{split}
\mathbb{E}_N | _{ \mathcal{A}_{ \Lambda}} = \mathbb{E} _N^\Lambda
		\end{split}
	\end{equation}
Furthermore, by the definition (\ref{eq:cond-exp}) of the finite-volume maps, $\mathbb{E}_N(A^*) = \mathbb{E}_N(A)^*$ for all $A \in \mathcal{A}$ and if $M \in\mathbb{N}$ and $M \geq N$, 
	\begin{equation}
		\begin{split}
\mathbb{E}_M \mathbb{E}_N = \mathbb{E}_N \mathbb{E}_M = \mathbb{E} _ N . 
		\end{split}
	\end{equation}
The family $\{ \mathbb{E}_N \}$ provides local approximations of quasi-local observables. For completeness, we record this as the following proposition and refer to \cite{NSY} for the proof. 

\begin{prop}\label{NSY44}
Let $\varepsilon \geq 0$. Suppose $A \in \mathcal{A}$ is such that for all $B \in \bigcup_{\substack{X \in \mathfrak{S}_{\mathbb{Z}^\nu}\\ X \cap \Lambda_N = \emptyset}} \mathcal{A}_X$:
	\begin{equation}
		\begin{split}
		\lVert [A, B] \rVert \leq \varepsilon \lVert B \rVert.
		\end{split}
	\end{equation}
Then $\lVert A - \mathbb{E}_N (A) \rVert \leq 2\varepsilon$. 
\end{prop}

\begin{proof}
See Corollary 4.4 of \cite{NSY}. 
\end{proof}

\section{Properties of $\caD_f$}\label{dfsec}
The map $\lV\cdot\rV_f:\caD_f\to \bbR_{\ge 0}$ is a norm on
$\caD_f$.
Note that $\lV A^* \rV_f=\lV A\rV_f$, and $\lV \bbE_N(A)\rV_f\le\lV A\rV_f$.
Furthermore, if $\sup_{N\in\nan} \frac{f(N)}{g(N)}<\infty$,
then $\caD_f\subset \caD_{g}$.
\begin{lem}\label{alg}
Let $f:(0,\infty)\to (0,\infty)$ be a continuous decreasing function 
with $\lim_{t\to\infty}f(t)=0$.
The set $\caD_f$ is a $*$-algebra which is a Banach space with respect to the norm $\lV\cdot\rV_f$.
\end{lem}
\begin{proof}
That $\caD_f$ is $*$-closed is trivial from $\lV A^*\rV_f=\lV A\rV_f$.
To see that $\caD_f$ is closed under multiplication, let
 $A,B\in\caD_f$.
 For each  $N\in\nan$, we have
\begin{align}
	\begin{split}
&\lV
AB-\bbE_{N}(AB)
\rV
\le
\lV
\lmk A-\bbE_{N}\lmk A\rmk\rmk\cdot  B\rV+
\lV
-\bbE_{N}\lmk
\lmk
A-\bbE_{N}\lmk A\rmk
\rmk \cdot B
\rmk
\rV
+\lV
\bbE_{N}\lmk A\rmk\cdot\lmk
B-\bbE_{N}\lmk B\rmk
\rmk
\rV\\
&
\le
\lmk 2\lV A\rV_f\lV B\rV+\lV A\rV\lV B\rV_f\rmk f(N)
\le 3\lV A\rV_f\lV B\rV_f f(N).
	\end{split}
\end{align}
Hence we obtain $AB\in \caD_f$, and $\caD_f$ is closed under the multiplication.

To prove that $\caD_f$ is complete with respect to $\lV \cdot\rV_f$,
let $\{A_n\}_n$ be a Cauchy sequence in $\caD_f$ with respect to $\lV\cdot\rV_f$.
As $\{A_n\}_n$ is Cauchy with respect to $\lV \cdot\rV$ as well, there is an
$A\in\caA$ such that $\lim_{n\to\infty}\lV A-A_n\rV=0$.
This $A$ belongs to $\caD_f$ because
\begin{align}
\sup_{N\in\nan}\frac{\lV
A-\bbE_{N}(A)
\rV}
{f(N)}
=\sup_{N\in\nan}
\lmk
\lim_{M\to \infty} 
\frac{\lV
A_M-\bbE_{N}(A_M)
\rV}
{f(N)}
\rmk
\le\sup_{M}\lV A_M\rV_f<\infty.
\end{align} 
Furthermore, we have
\begin{align}
\sup_N\frac{\lV
A-A_m-\bbE_{N}(A-A_m)
\rV}
{f(N)}
=\sup_N\lim_{n\to\infty}\lmk
\frac{\lV
A_n-A_m-\bbE_{N}(A_n-A_m)
\rV}
{f(N)}
\rmk
\le
\limsup_{n\to\infty} \lV A_n-A_m\rV_f.
\end{align}
Therefore, $A_m$ converges to $A\in\caD_f$ in $\lV \cdot\rV_f$-norm.
\end{proof}
\begin{lem}Let $f:(0,\infty)\to (0,\infty)$ be a continuous decreasing function 
with $\lim_{t\to\infty}f(t)=0$ with $M\in\nan$.
For any $A\in\caD_f$ and $B\in \caA_{\Lambda_M}$ and $M\in\nan$
we have
\begin{align}
\lV
BA
\rV_f
\le \lmk 1+\max\left\{
\frac 2{f(M)},1
\right\}\rmk
\lV B\rV \lV A\rV_f.
\end{align}
\end{lem}
\begin{proof}
This follows from the following inequality:
\begin{align}
	\begin{split}
&\lV
BA-\bbE_N(BA)
\rV\\
&\le \left\{
\begin{gathered}
2\lV B \rV \lV A\rV,\quad N\le M,\\
\lV
B\lmk A-\bbE_N(A)\rmk
\rV,\quad N>M.
\end{gathered}
\right.
\\
&\le \left\{
\begin{gathered}
2\lV B \rV \lV A\rV,\quad N\le M,\\
\lV B\rV \lV A\rV_f f(N)\quad N>M.
\end{gathered}
\right.
	\end{split}
\end{align}
\end{proof}
\begin{lem}\label{n47}
Let $f,f_1:(0,\infty)\to(0,\infty)$ be continuous decreasing functions.
Suppose that
and
\begin{align}
\lim_{N\to\infty}\frac{f(N)}{f_1(N)}=0.
\end{align}

Then we have
\begin{align}\label{215}
\lim_{M\to \infty}\lV A- \bbE_{M}(A)\rV_{f_1}=0,\quad A\in\caD_{f}.
\end{align}
\end{lem}
\begin{proof}
Let $A\in\caD_f$.
By the definition of $\caA$, we have
$\lim_{M\to \infty}\lV A- \bbE_{M}(A)\rV=0$.
We note that for $N\in\nan$,
\begin{align}
	\begin{split}
&\frac{\lV A- \bbE_{M}(A)-\bbE_N\lmk
A- \bbE_{M}(A)
\rmk\rV}{f_1(N)}
=\left\{
\begin{gathered}
\frac{\lV A-\bbE_N\lmk
A
\rmk\rV}{f_1(N)},\quad M\le N,\\
\frac{\lV A- \bbE_{M}(A)\rV}{f_1(N)}
,\quad M> N,
\end{gathered}
\right.\\
&=
\left\{
\begin{gathered}
\frac{\lV A-\bbE_N\lmk
A
\rmk\rV}{f(N)}\frac{f(N)}{f_1(N)},\quad M\le N,\\
\frac{\lV A- \bbE_{M}(A)\rV}{f(M)}\frac{f(M)}{f_1(N)}
,\quad M> N\\
\end{gathered}
\right.\\
&\le \lV A\rV_{f}
\left\{
\begin{gathered}
\frac{f(N)}{f_1(N)},\quad M\le N,\\
\frac{f(M)}{f_1(M)}
,\quad M> N\\
\end{gathered}
\right.\\
&\le\lV A\rV_{f}
\sup_{M\le L\in\nan}\lmk \frac{f(L)}{f_1(L)}\rmk
 \to 0,\; M\to\infty.
 	\end{split}
\end{align}
Hence we obtain (\ref{215}).
\end{proof}

\end{document}